\documentclass[sn-mathphys,Numbered]{sn-jnl}

\usepackage{amsmath,amssymb,amsfonts,amsthm}
\usepackage[utf8]{inputenc}
\usepackage{float}
\usepackage{graphics}
\usepackage{xspace}
\usepackage{caption}
\usepackage{comment}
\usepackage{wrapfig}

\usepackage{orcidlink}

\usepackage{graphicx}
\usepackage{tikz}
\usetikzlibrary{calc}

\usepackage{pgfplots}

\usepackage[boxruled,vlined,linesnumbered]{algorithm2e}

\newcommand{\bigOh}{\mathop{\mathcal{O}}\nolimits}  % big O
\newcommand{\nth}[1]{\text{${#1}$-th}}              % n-th

\newcommand{\f}[1]{\text{\texttt{#1}}\xspace}        % field
\newcommand{\M}{\ensuremath{M}\xspace}               % size of hashmap
\newcommand{\m}{\ensuremath{m}\xspace}               % \log_2 \M
\newcommand{\free}{\ensuremath{\varepsilon}\xspace}  % number of free slots
\newcommand{\del}{\ensuremath{\delta}\xspace}        % number of deleted slots
\newcommand{\used}{\ensuremath{\nu}\xspace}        % number of used slots
\newcommand{\N}{\ensuremath{N\xspace}}               % = \used when no   deletions
\newcommand{\lastfree}{\texttt{last\_free}\xspace}   % index to last free
\newcommand{\T}{\ensuremath{\mathcal{T}}\xspace}     % name of the hashmap
\newcommand{\call}[2]{\textsc{#1}({#2})}             % call to procedures 
\newcommand{\nil}{\textbf{nil}\xspace}               
\newcommand{\pair}[2]{\ensuremath{\langle #1, #2\rangle}}
\newcommand{\cpair}[2]{\ensuremath{#1,\ #2}}

\newcommand{\lua}{Lua\xspace} %in the middle of a sentence
\newcommand{\Lua}{Lua\xspace} %capitalized to start a sentence or for use in section titles

\newcommand{\bin}{\texttt{Bin}}

\usepackage{marginnote}
\newcommand{\reviewer}[1]{ %\marginpar{\it\footnotesize #1}%
\relax
}

\usepackage{thmtools,thm-restate}
\usepackage{hyperref}

\usepackage{mathrsfs}%
\usepackage{textcomp}%
\usepackage{manyfoot}%
\usepackage{booktabs}%
\usepackage{listings}%

\usepackage{orcidlink}

\lstdefinestyle{myLuastyle}
{
  language         = {[5.2]Lua},
  basicstyle       = \ttfamily,
  showstringspaces = false,
  upquote          = true,
  keywordstyle=\bfseries\color{blue}
}

\newtheorem{theorem}{Theorem}
\newtheorem{lemma}[theorem]{Lemma}
\newtheorem{proposition}[theorem]{Proposition}
\newtheorem{corollary}[theorem]{Corollary}
\newtheorem{remark}{Remark}
\newtheorem{example}{Example}

\begin{document}

\title[Analysis of Lua Hybrid Tables]{Mathematical Models to Analyze \Lua Hybrid Tables and Why They Need a Fix}

\author[1]{\fnm{Conrado} \sur{Martínez}\orcidlink{0000-0003-1302-9067}}\email{conrado@cs.upc.edu}

\author[2]{\fnm{Cyril} \sur{Nicaud}\orcidlink{0000-0002-8770-0119}}\email{cyril.nicaud@univ-eiffel.fr}

\author*[2]{\fnm{Pablo} \sur{Rotondo}\orcidlink{0000-0001-8777-1278}}\email{pablo.rotondo@univ-eiffel.fr}

\affil*[1]{\orgdiv{Dept.\ of Computer Science}, \orgname{Universitat Politècnica de Catalunya}, \orgaddress{\street{Campus Nord, Omega--241, Jordi Girona, 1--3}, \city{Barcelona}, \postcode{08034}, \country{Spain}}}

\affil[2]{\orgdiv{Laboratoire d'Informatique Gaspard-Monge (LIGM) UMR 8049}, \orgname{Université Gustave Eiffel}, \orgaddress{\street{Bâtiment Copernic, 5, boulevard Descartes, Cité Descartes}, \city{Champs-sur-Marne}, \postcode{77454}, \country{France}}}

\abstract{\Lua (Ierusalimschy \emph{et al.}, 1996) 
is a well-known scripting language, 
popular among many programmers, most notably in the gaming industry. Remarkably,  the only data-structuring mechanism in \lua are associative arrays, called  \emph{tables}. With \lua 5.0, the reference implementation of \lua introduced \emph{hybrid tables} to implement tables using both a hashmap and a dynamically growing array combined together: the values associated with integer keys are stored in the array part, when suitable, everything else is stored in the hashmap. All this is transparent to the user, who gets a unique simple interface to handle tables. In this paper we carry out a theoretical analysis of the performance of \lua's tables, by considering various worst-case and probabilistic scenarios. In particular, we uncover
some problematic situations for the simple probabilistic model where we add a new key with some fixed probability $p>\frac12$ and delete a key with probability $1-p$: the cost of performing $T$ such operations is proved to be $\Omega(T\log T)$ with high probability, where linear complexity is expected instead.}

\keywords{Algorithm engineering, Data structures, Probabilistic analysis of algorithms, Hash  tables, Lua}

\maketitle

\section{Introduction}
\label{sec:introduction}
When implementing the standard algorithms and data structures of a new programming language, engineers usually follow the classical solutions that have been validated by both practice and theory. Sometimes, however, they innovate and propose new ideas that fit best with the internal implementation of the language or that behave better with the typical data of their intended audience. For example, this was the case for the sorting algorithm \textsc{TimSort} used in the main implementation of \textsc{Python}.\footnote{\url{https://github.com/python/cpython/blob/main/Objects/listsort.txt}} This new, elegant and powerful sorting algorithm was quickly adopted by several other languages, including \textsc{Java}.\footnote{\url{https://docs.oracle.com/javase/7/docs/api/java/util/Arrays.html}} 
After a decade of existence, computer scientists started analyzing its efficiency,  helped to fix some issues and confirmed its excellent performances at a theoretical level~\cite{AugerJNP18,BussK19,GoRoBoBuHa15}. 
The actual principal implementation of the language \lua revisits the way maps, i.e., associative arrays, are structured internally, in a new and innovative structure named \emph{table}.  Studying this novel way of structuring data from a theoretical point of view is the main purpose of this article.

\Lua\footnote{From the Portuguese ``\emph{lua}'', meaning moon.}
is a scripting programming language~\cite{IFF96} created in the early nineties and adopted by many programmers, especially for the development of gaming applications;  it  keeps a base of tens of thousands of users worldwide. Like many scripting languages, \lua is characterized by its simplicity and extensibility, with the aim to help integrate code written in different programming languages.

The only data-structuring mechanism in \lua are so-called tables: this was a design decision which made the language simple, yet flexible and powerful. 
If $H$ is a \lua table then the assignment
\texttt{H[x] = y}
associates the value $y$ to the key $x$, 
for whatever $x$ and $y$, and regardless of their types. 
If $x$ was already a key in $H$, then the 
value associated to $H$ is updated and changed
to $y$. If $x$ is not present, then the instruction inserts the pair $\pair{x}{y}$ in the table $H$. The expression $H[x]$ actually returns a reference to the place where the value associated to $x$ is stored or the special
value \nil if the key is not present in the table; the reference can be then used to obtain the sought value or to assign a new value. To ``delete'' a pair $\pair{x}{y}$ from $H$ it is enough to assign \nil to $H[x]$. Everything is transparent to the user, including how the table grows to accommodate more and more pairs, or how unused space is restored back to the memory heap for future use.

Until \lua 4.0, tables were implemented strictly as hashmaps: all pairs $\pair{x}{y}$ were explicitly stored in a single 
hashmap, irrespective of the type of the keys $x$.
\Lua 5.0 brought on a new implementation of the tables in order to optimize their use as arrays: pairs with integer keys are stored in a separate actual array, without storing their keys,
provided that the index (the integer key) falls within the 
current range $[1,\ldots,n]$ of the array~\cite{IFF05}. The value $n$ changes dynamically so that the array always contains $> n/2$ non-nil values. All other pairs, when the key is not an integer, or it is outside the current range of the array, are stored in the hashmap as usual. The new \lua (hybrid) tables thus have two parts, called the hash-part or \emph{hashmap},
and the array-part or \emph{array}.

\noindent\textbf{Main contributions \& Plan of the paper.} In this paper, we provide a theoretical 
analysis covering some aspects of the performance of \lua hybrid tables.\footnote{All detailed descriptions in the remaining of the paper and our analysis refer to version 5.4.4 of \Lua (the most recent).} We first consider the implementation of hashmaps (Section~\ref{sec:hashmaps}). It is quite direct to establish that, in presence of both insertions and deletions, \Lua's hashmap is suboptimal in the worst case (amortized). This is due to incorrect resizing parameter settings of the hashmap, which can be exploited to construct a sequence of insertions and deletions that frequently trigger a recalculation of the entire hashmap. 

However, this worst case instance is very specific, and one can legitimately argue that it is very unlikely to happen in practice. Our main contribution is to propose a simple yet reasonable probabilistic scenario for insertions and deletions in a data structure, then establish that under this model, an operation in a \Lua hashmap takes $\Omega(\log n)$ time. This is done in Section~\ref{subsec:analysis-hashmaps-avg-case}, and this strongly advocates for changes in the implementations, which are discussed in the conclusion.

After that, in Section~\ref{sec:hybrid-tables}, we investigate the performance of the hybrid table as a whole, this time focusing on sequences of insertions involving integer keys, which should exploit, as much as possible, 
the array component of \lua's hybrid tables. We show first that a carefully crafted sequence of $n$ insertions will require super-linear cost (Example~\ref{example:hybrid-1} and Proposition~\ref{worst-case-permutations-1}). We also show that less adverse scenarios, in particular, some that may arise naturally in practice, will need expected constant amortized time per insertion, however the array part will be empty most of the time (Lemma~\ref{lemma:half-set-permutation-2nd-moment}), and thus the advantages of the hybrid scheme can become blurred (Theorem~\ref{thm:average-case-permutations2}). 

Finally, in the conclusions (Section~\ref{sec:conclusions}), we  present some striking experiments, comparing the performance of \Lua with a \Lua modified according to a classical solution to avoid the drawback of the hybrid table's implementation. 

A preliminary version of this paper was presented in a conference extended abstract~\cite{MNR22}.

\section{Hashmaps in \Lua}
\label{sec:hashmaps}

\subsection{Description of the hashmap algorithms}
\label{subsec:hashmaps-algorithms}
In this section, we give a precise description of the hashmap algorithms in \Lua, corresponding to the version~5.4.\footnote{The source code can be found in the file \url{https://www.lua.org/source/5.4/ltable.c.html}}  

When non-empty, a hashmap in \lua consists of an array of size $\M=2^\m$ for some $\m\ge 0$. Each slot contains a \emph{key}, a \emph{value} and the index \emph{next} of the next location to probe in the search sequence if the sough key is not the one at that slot (\f{next} is \nil if there is no successor). When both the key and the value are \nil, the slot is empty; for slots that have been deleted the value has been set to \nil, but the key is retained; the slots that contain the actual elements of the table have both their keys and their value fields non-\nil. 
Throughout the article, a slot is said to be \emph{used}, \emph{deleted} or \emph{free} when it contains a pair $\pair{\text{key}}{\text{value}}$, 
a pair $\pair{\text{key}}{\nil}$ or a pair $\pair{\nil}{\nil}$, respectively.

In addition to the array, the data structure also keeps an index \lastfree pointing to the first slot that must be checked when looking for a free slot in a downwards scanning of the array.
Initially, we set $\lastfree$ to $\M-1$, and it can only decrease.

\noindent\textbf{Search.} The search for a key $x$ simply consists
in computing its hash value then following the \f{next} links until we find the key $x$ and return the associated value (success if it is not \nil) or the end of the list (failure). Notice that returning the value of the last probed slot, we will return \nil  whenever $x$ is not in the table (either because $x$ was deleted or because we reach a free slot).

\noindent\textbf{Deletion.} To delete a key $x$, we search for it. The associated value of the slot where the search ends is set to \nil (it might already be \nil if $x$ was not in the table). 
The \f{next} field remains unaltered to maintain the chaining.

\noindent\textbf{Insertion.}
If one wants to set $x\mapsto y$ and the key $x$ is already in the table, even with a deleted status, we just update its associated value to $y$. If $x$ is not already there, let $i$ be the position corresponding to the hash value of $x$ (the \emph{main position} of $x$, in \lua parlance). If the slot $i$ is free or deleted, the key $x$ and its value $y$ are put there, with a \nil \f{next}-link when the slot is free and keeping the \f{next}-link value if it is deleted.
Otherwise, there is a collision with another pair $\pair{x'}{y'}$ ($y'\not=\nil$) at the position $i$. There are two distinct cases when it happens. If $x'$ is at its main position, a free slot is found for $\pair{x}{y}$ and the chaining is updated so that $\pair{x}{y}$ is on the second position of the linked list starting at index $i$. This is easily done by updating the two \f{next}-links at index $i$ and at the free slot.
If $x'$ is not at its main position, it is moved to a free slot, which requires a list scan starting from the hash value of $x'$ to find its predecessor in its chaining. Then $\pair{x}{y}$ is set at index $i$, with a \nil \f{next}-link, starting a new chain within the table.

Looking for a free slot is accomplished by scanning the table from right to left, starting at position \lastfree, until a free slot is found. Importantly, \emph{deleted slots are ignored during this process} in \lua, to avoid problems with the chaining. If no free slot is found, i.e., \lastfree exits the left boundary of the array, then a rehash occurs: the number of used keys $n$ is determined, then a new hashmap of size $2^m$ is allocated, where $m$ is the smallest integer such that $n+1\leq 2^m$. All pairs $\pair{\text{keys}}{\text{value}}$ and the pair $\pair{x}{y}$ are then inserted into the new map.

So in any case, just after its insertion,  $\pair{x}{y}$ is at the first or second place in its chain. The algorithms are depicted in Fig.~\ref{algo:hashmap} and examples of insertions in a \lua hashmap are depicted in Fig.~\ref{fig:insertions}.

\begin{figure}[H]
\begin{algorithm}[H]\small
\DontPrintSemicolon
  $i\gets \call{Hash}{x}$\;
  \lIf{$\T[i].\f{value}=\nil$}{
  	$(\T[i].\f{key},\ \T[i].\f{value})\gets (\cpair{x}{y})$
  }
  \Else{
  	$f\gets \call{GetFreePos}{\T}$\;
	\lIf{$f=\nil$}{
		$\T\gets\call{ReHash}{\T,x,y}$
	}
	\Else{
		$j\gets \call{Hash}{\T[i].\f{key}}$\;
		\If{$i=j$}{
			$(\T[f].\f{key},\ \T[f].\f{value})\gets (\cpair{x}{y})$\;
			$\T[f].\f{next} \gets \T[i].\f{next}$\;
			$\T[i].\f{next} \gets f$\;
		}
		\Else{
			$p \gets \call{Predecessor}{\T,i}$\;
			$\T[f]\gets \T[i]$\;
			$\T[p].\f{next}\gets f$\;
			$(\T[i].\f{key},\ \T[i].\f{value},\ \T[i].\f{next})\gets (\cpair{x}{y},\ \nil)$\;
		}
	}
  }
  \caption{ $\call{Insertion}{\T,x,y}$} 
  \label{algo:insertion}
\end{algorithm}
\begin{algorithm}[H]
\small
\DontPrintSemicolon
  $n,\ M \gets 0,\ 1$\;
  \For{$i$ index of $\T$}{
  	\lIf{$\T[i]$.\f{value}$\neq$\,\nil}{
		$n\gets n+1$
	}
  }
  \lWhile{$\M<n$}{$\M\gets 2\M$}
   $\T'\gets\call{Allocate}{\M}$\;
   \For{$i$ index of $\T$}{
  	\lIf{$\T[i]$.\f{value} $\neq$ \nil}{
		$\call{Insertion}{\T',\T[i].\f{key},\T[i].\f{value}}$
	}
  }
  $\call{Insertion}{\T',x,y}$\;
  \Return $\T'$\;
\caption{ $\call{ReHash}{\T,x,y}$} 
   \label{algo:rehash}
\end{algorithm}
\begin{algorithm}[H]
\small
\DontPrintSemicolon
  \While{$\T.\lastfree\geq 0$ and $\T[\T.\lastfree].\f{key}\neq \nil$}{
  	$\T.\lastfree\gets\T.\lastfree-1$
  }
  \lIf{\T.\lastfree$=-1$}{
  	\Return \nil
  }
  \lElse{
  	\Return \T.\lastfree
  }
  \caption{ $\call{GetFreePos}{\T}$}
  \label{algo:getfreepos}
\end{algorithm}
\caption{Algorithm for the insertion of a pair key/value $\pair{x}{y}$ in a \lua hashmap (instruction 
$\text{\texttt{\T\!\!\![x]=y}}$ in \lua), when $x$ is not already a key in $\T$. Predecessor just iterates through the chaining of $\T[i]$ to find its predecessor.\label{algo:hashmap}} 
\end{figure}

\noindent
\begin{minipage}[t]{.48\textwidth}
    \vspace{0pt}
	\begin{tikzpicture}[xscale=.7,yscale=1]
		\path (0,0) -- (0,-9.9);
	
	%% first hashmap
	\node (i) at (0,0) {};
	\draw[thick] (i) rectangle ++ (8,1);
	\foreach \i in {1,2,...,8} {
		\draw[thick] ($(i) + (\i,0)$) -- ++ (0,1);
	}
	\foreach \i in {0,1,...,7} {
		\node at ($(i) + (\i+.5,1.3)$) {\footnotesize $\i$};
	}
	\foreach \i/\k/\v/\n in {2/k_1/v_1/n1, 7/k_2/v_2/n2, 5/k_3/v_3/n3,6/\ell/w/m1} {
		\node at ($(i)+(\i+.5,.75)$) {\small $\k$};
		\node at ($(i)+(\i+.5,.4)$) {\small $\v$};
		\node[inner sep=3pt,minimum size=0pt]  (\n) at ($(i)+(\i+.5,.12)$) {$\bullet$};
	}
%	\draw (n1) edge[->,bend right = 50] (n4);
%	\draw (n4) edge[->,bend right = 80] (n3);
	\draw (n1) edge[->,bend right = 40] (n3);
	
	\draw (n3) edge[->,bend right = 50] (n2);
	\draw[->,thick] ($(i)+(5.5,1.8)$) -- node[right]{\scriptsize\tt last\_free}++ (0,-.3);
	
	\draw[blue,->,very thick] (.4,-1) -- node[black,right,align=left]{\footnotesize insert key $k_4$ (hash=2) with value $v_4$\\\footnotesize collision with $k_1$, in its main position} ++ (0,-1.4);
	
	%% second hashmap
	\node (i) at (0,-4.5) {};
	\draw[thick] (i) rectangle ++ (8,1);
	\foreach \i in {1,2,...,8} {
		\draw[thick] ($(i) + (\i,0)$) -- ++ (0,1);
	}
	\foreach \i in {0,1,...,7} {
		\node at ($(i) + (\i+.5,1.3)$) {\footnotesize $\i$};
	}
	\foreach \i/\k/\v/\n/\c in {2/k_1/v_1/n1/black, 7/k_2/v_2/n2/black, 5/k_3/v_3/n3/black,4/k_4/v_4/n4/blue,6/\ell/w/m1/black} {
		\node[\c] at ($(i)+(\i+.5,.75)$) {\small $\k$};
		\node[\c] at ($(i)+(\i+.5,.4)$) {\small $\v$};
		\node[\c,inner sep=3pt,minimum size=0pt]  (\n) at ($(i)+(\i+.5,.12)$) {$\bullet$};
	}
	\draw (n1) edge[blue,->,thick,bend right = 50] (n4);
	\draw (n4) edge[blue,->,thick,bend right = 80] (n3);
	\draw (n1) edge[red,->,thick,bend right = 43,dotted] (n3);
	
	\draw (n3) edge[->,bend right = 50] (n2);
	\draw[blue,->,thick] ($(i)+(4.5,1.8)$) -- node[right]{\scriptsize\tt last\_free}++ (0,-.3);
	
	\draw[blue,->,very thick] (.4,-5.5) -- node[black,right,align=left]{\footnotesize insert key $x$ (hash=4) with value $y$\\\footnotesize collision with $k_4$, not in  main position} ++ (0,-1.4);
	
	%% third hashmap
	\node (i) at (0,-9) {};
	\draw[thick] (i) rectangle ++ (8,1);
	\foreach \i in {1,2,...,8} {
		\draw[thick] ($(i) + (\i,0)$) -- ++ (0,1);
	}
	\foreach \i in {0,1,...,7} {
		\node at ($(i) + (\i+.5,1.3)$) {\footnotesize $\i$};
	}
	\foreach \i/\k/\v/\n/\c in {2/k_1/v_1/n1/black, 7/k_2/v_2/n2/black, 5/k_3/v_3/n3/black,3/k_4/v_4/n4/blue,6/\ell/w/m1/black, 4/x/y/nx/blue} {
		\node[\c] at ($(i)+(\i+.5,.75)$) {\small $\k$};
		\node[\c] at ($(i)+(\i+.5,.4)$) {\small $\v$};
		\node[\c,inner sep=3pt,minimum size=0pt]  (\n) at ($(i)+(\i+.5,.12)$) {$\bullet$};
	}
	\draw (n1) edge[blue,->,thick,bend right = 50] (n4);
	\draw (n4) edge[->,bend right = 50] (n3);
	\draw (n1) edge[red,->,thick,bend right = 43,dotted] (nx);
	
	\draw (n3) edge[->,bend right = 50] (n2);
	\draw[blue,->,thick] ($(i)+(3.5,1.8)$) -- node[right]{\scriptsize\tt last\_free}++ (0,-.3);

	\end{tikzpicture}
	\captionof{figure}{If there is no deleted value and if there is a collision with a key that is in its main position, the new inserted element is placed in a free spot and is at the second position in its chaining. If the colliding key is not at its main position, the new element is put there, and the colliding element is placed in the free spot.}
	\label{fig:insertions}
\end{minipage}
\hfill
\begin{minipage}[t]{.48\textwidth}
    \vspace{0pt}
	\begin{tikzpicture}[xscale=.7,yscale=1]
	\path (0,0) -- (0,-9.9);
	%% first hashmap
	\node (i) at (0,0) {};
	\draw[thick] (i) rectangle ++ (8,1);
	\foreach \i in {1,2,...,8} {
		\draw[thick] ($(i) + (\i,0)$) -- ++ (0,1);
	}
	\foreach \i in {0,1,...,7} {
		\node at ($(i) + (\i+.5,1.3)$) {\footnotesize $\i$};
	}
	\foreach \i/\k/\v/\n in {2/k_1/v_1/n1, 7/k_2/\--/n2, 5/k_3/\--/n3,6/\ell/w/m1/,4/z/\--/nz} {
		\node at ($(i)+(\i+.5,.75)$) {$\k$};
		\node at ($(i)+(\i+.5,.4)$) {$\v$};
		\node[inner sep=3pt,minimum size=0pt]  (\n) at ($(i)+(\i+.5,.12)$) {$\bullet$};
	}
	\draw (n1) edge[->,bend right = 40] (n3);	
	\draw (n3) edge[->,bend right = 50] (n2);
	\draw[->,thick] ($(i)+(5.5,1.8)$) -- node[right]{\scriptsize\tt last\_free}++ (0,-.3);
	
	\draw[blue,->,very thick] (.4,-1) -- node[black,right,align=left]{\footnotesize insert key $x_1$ (hash=7) with value $y_1$\\\footnotesize collision with  $k_2$ having no value} ++ (0,-1.4);
	
	%% second hashmap
	\node (i) at (0,-4.5) {};
	\draw[thick] (i) rectangle ++ (8,1);
	\foreach \i in {1,2,...,8} {
		\draw[thick] ($(i) + (\i,0)$) -- ++ (0,1);
	}
	\foreach \i in {0,1,...,7} {
		\node at ($(i) + (\i+.5,1.3)$) {\footnotesize $\i$};
	}
	\foreach \i/\k/\v/\n/\c in {2/k_1/v_1/n1/black, 7/x_1/y_1/nx1/blue, 5/k_3/\--/n3/black,6/\ell/w/m1/black/,4/z/\--/nz//black} {
		\node[\c] at ($(i)+(\i+.5,.75)$) {\small $\k$};
		\node[\c] at ($(i)+(\i+.5,.4)$) {\small $\v$};
		\node[inner sep=3pt,minimum size=0pt]  (\n) at ($(i)+(\i+.5,.12)$) {$\bullet$};
	}
	\draw (n1) edge[->,bend right = 40] (n3);	
	\draw (n3) edge[->,bend right = 50] (nx1);
	\draw[->,thick] ($(i)+(5.5,1.8)$) -- node[right]{\scriptsize\tt last\_free}++ (0,-.3);
	
	\draw[blue,->,very thick] (.4,-5.5) -- node[black,right,align=left]{\footnotesize insert key $x_2$ (hash=7) with value $y_2$\\\footnotesize collision with $x_1$, in its main position} ++ (0,-1.4);
	
	%% third hashmap
	\node (i) at (0,-9) {};
	\draw[thick] (i) rectangle ++ (8,1);
	\foreach \i in {1,2,...,8} {
		\draw[thick] ($(i) + (\i,0)$) -- ++ (0,1);
	}
	\foreach \i in {0,1,...,7} {
		\node at ($(i) + (\i+.5,1.3)$) {\footnotesize $\i$};
	}
	\foreach \i/\k/\v/\n/\c in {2/k_1/v_1/n1/black, 7/x_1/y_1/nx1/black, 5/k_3/\--/n3/black,3/x_2/y_2/nx2/blue,6/\ell/w/m1/black, 4/z/\--/nz/black} {
		\node[\c] at ($(i)+(\i+.5,.75)$) {$\k$};
		\node[\c] at ($(i)+(\i+.5,.4)$) {$\v$};
		\node[\c,inner sep=3pt,minimum size=0pt]  (\n) at ($(i)+(\i+.5,.12)$) {$\bullet$};
	}
	\draw (n1) edge[->,bend right = 40] (n3);	
	\draw (n3) edge[->,bend right = 50] (nx1);
	
	\draw (nx1) edge[blue, ->,bend left = 35] (nx2);
	
	\draw[blue,->,thick] ($(i)+(3.5,1.8)$) -- node[right]{\scriptsize\tt last\_free}++ (0,-.3);

	\end{tikzpicture}
	\captionof{figure}{If the main position of the newly inserted key has a deleted value, it is just put there. If it is used, then the insertion proceeds as previously, considering the deleted spot as occupied for the search of a free spot. Observe that at the end, the hash values of 2 and 7 share the same chain, which could not happen with no deletion.}
	\label{fig:insertions with deleted}
\end{minipage}
\medskip

\reviewer{The insertion algorithm and other auxiliary functions in pseudo-code can be found in Appendix~\ref{sec:pseudocode}.}

The algorithms of~\autoref{algo:hashmap} correspond to the %function in the cited 
{\tt C} implementation as follows: the algorithm
{\tt Insertion} corresponds to {\tt luaH\_newkey}, the algorithm {\tt GetFreePos} corresponds to {\tt getfreepos}, and {\tt ReHash} to {\tt rehash}. More precisely, {\tt luaH\_newkey}  performs the insertion of a new element, while the function {\tt  luaH\_set} corresponds to an assignment \texttt{H[x] = y}, where the key $x$ could already be present in the table $H$ (the value is updated in this case).

\subsection{Settings and analysis when there is no deletion}
\label{subsec:analysis-hashmaps-no-deletions}
Designing accurate hash functions is a whole field on its own, and it is not the topic of this paper. So, throughout the article, we consider that for a hashmap of length $M$, the hash values of different keys are independent uniform random integers of $\{0,\ldots,\M-1\}$. They are sampled again when a rehash occurs. This is the standard assumption for such analysis that do not go into the details of specific hash functions~\cite{Knuth3}. This model is called the \emph{simple uniform hashing}.

In the absence of deletions, \lua's hashmaps behave  as \emph{separate chaining} hashing~\cite{Knuth3}. If $N$ is the number of elements in the table and $M$ the size of the table, the \emph{load factor} is classically defined as $\alpha=\N/\M \leq 1$.
 The average cost (measured as the number of slots inspected) of successful searches ($S_\N$) and unsuccessful searches ($U_\N$) is~\cite[\S 6.4, p. 525]{Knuth3}, as $N$ and $M$ tends to infinity: $
S_\N \approx 1+\frac{\alpha}2$ and 
$U_\N \approx 1+\frac{\alpha^2}2$. Note that it may seem strange that $U_\N \leq S_N$, but it is because of the implicit conditioning that successful searches do not consider free slots. It is worth mentioning here that if we used coalescing chaining, that is, adding a new pair $\pair{x}{y}$ always as successor of a conflicting pair $\pair{x'}{y'}$, whether that pair sits at its main location or not, then scanning the table to look for the \texttt{Predecessor} of $\pair{x'}{y'}$ would be unnecessary and the performance of successful and unsuccessful searches would degrade only slightly, namely to $S_\N\approx 1+\frac{e^{2\alpha}-1-2\alpha}{8\alpha}+\frac{\alpha}{4}$ and $U_\N\approx1+\frac{e^{2\alpha}-1-2\alpha}{4}$ (see~\cite[\S 6.4, p. 524]{Knuth3}). The improvement granted by pure separate chaining ``is not enough of an improvement over [coalesced chaining] to warrant changing the algorithm [=using separate chaining instead of coalesced chaining]''~\cite{Knuth3}.

Classically, the rehashing procedure by doubling the capacity has a constant amortized cost, as the pointer \lastfree cannot be decreased more than $\M$ times.
In conclusion, everything is well known when there is no deletion, and the expected amortized cost of an insertion is $\bigOh(1)$.

\subsection{An Unlikely Worst-Case Scenario When Considering Deletions}
\label{subsec:analysis-hashmaps-worst-case}
In this section, we establish that the situation changes significantly when we consider deletions, and that the expected amortized cost in not constant anymore. 
We estimate the time taken by the whole process by counting the number of times the function to insert a key is called: it is called once when an insertion is performed, unless a rehash occurs, in which case it is also called once for every key having an assigned value. Clearly, this number of calls is a lower bound for the complexity of the whole process. We suppose that keys are not integers, so only the hash-part is studied in this section.

\begin{example}[deletions-insertions in a full hashmap]
\label{worst-case-hash}
If we delete an element from a full hashmap of size $\M=2^m$ and then perform an insertion of a new element, we rehash the whole table into the same size and the hash is going to be full again. The only case it does not happen if when the hash value of the newly inserted key is equal to the position of the previously deleted key, but this is very unlikely (probability $\frac1M$).
Each rehash costs $\Theta(\M)$ calls
to the insertion function of \lua. If we keep this alternation of delete-insert for $M$ times, the cost is huge:  for $\Theta(M)$ operations, we obtain a quadratic cost $\Theta(\M^2)$. 
\end{example}

One can legitimately object that this scenario is too unlikely to question the implementation. Users normally alternate a number of insertions with deletions in a more complex pattern, and there is no reason it happens exactly when the table is full. In the next section we present a simple, yet natural, probabilistic model for insertions and deletions.

\section{Analysis Under a Probabilistic Model for Insertions/Deletions in Hashmaps}

Recall first that we use the classical simple uniform hashing assumption~\cite{Knuth3} to model the behavior of the hash function, in both our worst case and probabilistic analysis. Hence, there is a layer of randomness within the algorithms, induced by this model, as it is classical in the field of randomized algorithm.
In this section, we add probabilities on the \emph{inputs} of the algorithms, which is the first step to perform the average case analysis. This randomness is not within the algorithm, it is used to model the input. Hence, in some sense, we are proposing in the following an average case analysis of a randomized algorithm: there are two layers of randomness, within the algorithm and on the inputs.

Usually, such analysis is not necessary, as hashmaps behave efficiently in the usual worst case amortized analysis of randomized algorithms (the randomness is just used to model the hash function). Which trivially implies  the efficiency on average for any distribution on the inputs. 

The situation is different here, as we established the inefficiency for the worst-case amortized setting, and want to improve on this result by proving that the algorithm still have an expected $\Omega(\log n)$ running time by elementary operation in a reasonable probabilistic model on the inputs.

\subsection{Description of the Probabilistic Model}
For any fixed $p\in(\frac12, 1)$, our probabilistic model is the following:
\begin{itemize}
    \item We perform $T$ operations starting from an empty hashmap. 
    \item Each operation is an insertion with probability $p$ and a deletion with probability $1-p$, independently of the previous operations. There is one exception when the hashmap is empty, in which case an insertion is always performed (with probability $1$).
    \item An insertion consists in adding a new key. As we use the simple uniform hashing assumption, its associated hash value is a uniform random integer between $0$ and $M-1$, where $M$ is the current size of the hashmap. We only consider keys that are not integer, so we do not trigger the use of the array-part: we only study the hashmap in this section.
    \item A deletion consists in deleting a key already present in the hashmap, taken uniformly at random amongst the keys currently present in the hash-map. Recall that in Lua, only the value associated to the key is deleted during this process.
    \end{itemize}
When considering the $t$-th operation, for $t\in\{0,\ldots, T\}$, we will often write that we are ``at time $t$'', and freely use the subscript $t$ to design the current value of a parameter of the hashmap at time $t$. For instance, $M_t$ denote the size of the hashmap obtained after $t$ operations. Observe that in our settings, we have a probabilistic process, and $M_t$ is a random variable. 

The parameter $p$ is chosen greater than $\frac12$ so that the expected number of keys and the expected size of the hashmap increase linearly with the time. Taking $p<\frac12$ yields non-interesting cases, and $p=\frac12$ is a very specific, singular case, which could be studied using similar techniques. 

\subsection{Technical tools}
In this section, we introduce the main tool we use for our proofs: Hoeffding's inequality for binomial random variables.
For $\theta\in(0,1)$ and $s\in\mathbb{Z}_{>0}$, we denote by ${\tt Bin}(s,\theta)$ the binomial law defined for all $i\in\{0,\ldots, s\}$ by:
\[
\Pr\left({\tt Bin}(s,\theta)=i\right) =
\binom{s}{i} \theta^i(1-\theta)^{s-i}.
\]

The following is a particular case of a classical result in Probability Theory,
which gives exponential bounds for the probability of ${\tt Bin}(s,\theta)$ deviating from its expected value $s\theta$. See for instance~\cite[Thm 2.8]{concentrationinequalities}
\begin{proposition}[Hoeffding's inequality] %Chernoff bounds %\cite{Hoe63,MU05,MR95,Ros01}]
\label{prop:hoeffding-ineq}
Let $s\in \mathbb{Z}_{>0}$ and $\theta \in (0,1)$. Then for $t\in \mathbb{R}_{>0}$ we have the bounds:
\label{prop:hoeffding}
\begin{equation} \label{eq:hoeffding-1}
    \Pr\left({\tt Bin}(s,\theta)- s \theta \geq t \right) \leq \exp\left(-\frac{2t^2}{s}\right)\,,
\end{equation}
and
\begin{equation}\label{eq:hoeffding-2}
    \Pr\left(|{\tt Bin}(s,\theta)- s \theta| \geq t \right) \leq 2\exp\left(-\frac{2t^2}{s}\right)\,.
\end{equation}
\end{proposition}

\subsection{Main result and proof strategy}

Let us state our main result, which establishes the inefficiency of the \lua's implementation of hashmap for our probabilistic scenario. We use the following definition: a property holds with exponentially (resp.\ super-polynomially) high probability in $T$ when the probability that it does not hold is less than $\exp(-cT)$ (resp.\ $\exp(-T^{c})$) for some positive constant $c$ and $T$ sufficiently large.

\begin{theorem}
\label{thm:main proba hashmap}
Let $p\in(\frac12,1)$. Starting from an empty hashmap, if $T$ operations of insertions with probability $p$ and deletions with probability $1-p$ are performed, then with super-polynomially  high probability in $T$, the insertion function of \lua is called  $\Omega(T\log T)$ times. As a consequence, the expected running time of the whole process is in $\Omega(T\log T)$.
\end{theorem}

As we will see, this is mostly because \lua spends a lot of time rehashing almost full hashmaps without increasing their size,  impairing the efficiency of the data structure. This shows on our simulation, as depicted in Fig.~\ref{fig:experiment-rehash-total}. 

Informally, the proof strategy of Theorem~\ref{thm:main proba hashmap} is the following. With high probability, the number of keys in the hashmap after $T$ operation is linear in $T$. We can further prove that in the process, for a well-chosen $M=\Theta(T)$, with high probability, we allocate a hashmap of size $M$ for the first time at some time $t_{M}$ and of size $2M$ for the first time at some time $t_{2M}$ with $t_M\leq t_{2M}\leq T$.

Since keys are added one by one, at time $t_M$ the hashmap contains exactly $M/2-1$ empty slots. We will establish that for some positive constant $\gamma$, with high probability, if there are $m$ empty slots after a rehash of size $M$, then there are at least $\gamma m$ empty slots at the next rehash (the hashmap's size is therefore still $M$ after the rehash). This holds with super-polynomially high probability for $m\geq \sqrt{M}$. The intuitive reason behind this phenomenon is that during the process of adding $m$ new keys, some are deleted and sufficiently many slots they occupied are not reused by the insertion process.

Hence, with super-polynomially  high probability, at the rehash times starting from time $t_M$, the hashmap contains $M/2-1$ empty slots, then at least $\gamma(M/2-1)$ empty slots, then at least $\gamma^2(M/2-1)$ empty slots, etc. So we need a logarithmic number of rehashes before the number of empty slots becomes smaller than $\sqrt{M}$. Since each rehash costs $\Theta(M)$ re-insertions, the total cost is $\Omega(M\log M)$ time, which is $\Omega(T\log T)$.

To formalize this proof sketch, we have to  study the process and evaluate finely the error term, that is, the probability that the typical scenario describe above does not hold.

\begin{figure}
\begin{minipage}[t]{0.45\textwidth}
\vspace{.6cm} % TODO change as needed
        {
  \begin{tikzpicture}[xscale=0.65,yscale=0.55]
  \begin{axis}[
  xmax=1048576,xmin=1,ymin=0,
   xlabel={size of new hashmap $\M=2^m$},
    ylabel={number of rehashes},every y tick scale label/.append style={xshift=-1.6em},xmode=log,  log basis x={2},every x tick scale label/.append style={xshift=2.8em, yshift=1.4em}]
    
        \addplot  [green!50!black,mark=x]  coordinates {

(1,1.000000000)
(2,1.670000000)
(4,2.130000000)
(8,3.230000000)
(16,4.680000000)
(32,5.770000000)
(64,7.880000000)
(128,9.260000000)
(256,10.960000000)
(512,11.740000000)
(1024,13.620000000)
(2048,16.020000000)
(4096,17.930000000)
(8192,19.300000000)
(16384,20.850000000)
(32768,22.850000000)
(65536,24.360000000)
(131072,25.970000000)
(262144,27.870000000)
(524288,29.770000000)
(1048576,31.530000000)
%  p = 0.6
     };

    \addplot [blue,mark=x] coordinates {

(1,1.000000000)
(2,1.020000000)
(4,1.220000000)
(8,1.220000000)
(16,1.700000000)
(32,1.850000000)
(64,2.190000000)
(128,2.380000000)
(256,2.900000000)
(512,3.010000000)
(1024,3.370000000)
(2048,3.710000000)
(4096,4.090000000)
(8192,4.360000000)
(16384,4.690000000)
(32768,5.040000000)
(65536,5.390000000)
(131072,5.770000000)
(262144,5.950000000)
(524288,6.330000000)
(1048576,6.640000000)
(2097152,7.000000000)
(4194304,7.150000000)
%  p = 0.9
     };

    \addplot  [red,mark=x]  coordinates {
(1,1.000000000)
(2,1.280000000)
(4,1.470000000)
(8,1.870000000)
(16,2.470000000)
(32,3.230000000)
(64,3.630000000)
(128,4.100000000)
(256,4.700000000)
(512,5.540000000)
(1024,6.060000000)
(2048,6.620000000)
(4096,7.550000000)
(8192,8.080000000)
(16384,8.650000000)
(32768,9.380000000)
(65536,9.710000000)
(131072,10.680000000)
(262144,11.270000000)
(524288,11.750000000)
(1048576,12.340000000)
(2097152,12.990000000)
(4194304,13.830000000)
%  p = 0.75
     };

  \end{axis}
\end{tikzpicture}  

        }
\end{minipage}
\hfill
\begin{minipage}[t]{0.5\textwidth}
    \vspace{0cm}
    \caption{Total number of rehashes ($y$-axis) producing a hashmap of a given size $\M$ ($x$-axis) during a run of $T=10^7$ operations. The plot is logarithmic in $\M$($x$-axis). Plots correspond to $p=0.6$ (green), $p=0.75$ (red) and $p=0.9$ (blue). Each point is the average result of 100 simulations directly using \lua. For instance, if $p=0.75$ we get roughly $10$ rehash that produces a table of size $2^{16}$, where we would want $1$ (or $2$ if we are unlucky).
        \label{fig:experiment-rehash-total}
    }\end{minipage}
    \vspace{-.5cm}
\end{figure}

\subsection{Proof of the main result}
\label{subsec:analysis-hashmaps-avg-case}
This section is devoted to the proof of Theorem~\ref{thm:main proba hashmap}, which follow the general strategy described in the previous section. Let us first introduce some  useful notations.
At any time $t\in\{0,\ldots, T\}$, the hashmap has size $\M_t$, and contains $\free_t$ free cells, $\del_t$ deleted cells (keys with no values), and $\used_t$ used cells\footnote{We use the Greek letter $\used$ for consistency with the other notations in this section, $\used_t$ corresponds to $\N$ when no deletion occurred.} (keys with values). As every slot is in one of the three states, at any time $t$, we have the identity $\free_t+\del_t+\used_t=M_t$. Also, at time $t=0$ the hashmap is empty, and thus $\M_0=\free_0=\del_0=\used_0=0$.

As an insertion is performed with probability
$p>\frac12$ and a deletion with probability $1-p<\frac12$, the number of used keys in
the table increases by $2p-1>0$ in expectation with each
operation. This can be turned into a statement with exponentially high
probability using Hoeffding’s inequality (Proposition~\ref{prop:hoeffding-ineq}), as stated in the
following lemma.

\begin{lemma}\label{lm:random increase}
For any real constant $c$ such that $0<c<2p-1$, if at a given time $s$ there are $\used_s$ used slots in the hashmap, then after $t$ more operations $\used_{s+t}\geq \used_s+ct$ with exponentially high probability in $t$.
\end{lemma}

\begin{proof}
Let $I_t$ denote the number of insertions made
during the first $t$ operations. If the hashmap is never empty, one has $I_t\sim\bin(t,p)$.  Since operations that are not insertions are deletions in our model, the number of used slots at time $t$ is $2I_t-t$.  By Hoeffding Inequality, Equation~\eqref{eq:hoeffding-2}, this quantity is strongly concentrated around its expectation $2pt-t=(2p-1)t$. We deal with a possibly empty hashmap by observing that the real $I_t$ is stochastically lower bounded by $\bin(t,p)$, yielding the result.
 \end{proof}

Set $\beta = \frac{2p-1}3$. As we want to establish an asymptotic result in $T$, we assume in the sequel that $T$ is sufficiently large to establish the required properties. A first consequence of Lemma~\ref{lm:random increase} is that, with exponentially high probability in $T$, at some time the hashmap will have size \M with $\frac{\beta}2T< \M \leq \beta T$, and some time after that the  hashmap will reach size $2\M$. 

Let $t_h$ be any time when we just rehash into a hashmap of size $\M$, where $\M$ is the unique power of two such that  $\frac{\beta}2T< \M \leq \beta T$. It is not necessarily the first time we rehash into a hashmap of size $\M$. As we just rehashed, we have $\del_{t_h}=0$, and $\M = \M_{t_h} =\used_{t_h}+\free_{t_h}$.
As long as the hashmap is not empty 
and that there is no rehash 
we have for $t\geq t_h$:
\begin{equation}\label{eq:insertions-deletions}
	\del_{t+1} = 
	\begin{cases}
		\del_t-1 & \text{with probability }\frac{p\del_t}{M} \qquad \qquad [\textit{insertion in a deleted slot}], \\
		\del_t & \text{with probability }p\left(1-\frac{\del_t}{M}\right) \ \  \,[\textit{insertion in a free slot}],\\
		\del_t + 1 & \text{with probability }1-p\qquad \quad\,  [\textit{deletion}].
	\end{cases}
\end{equation}

Intuitively, by comparing $\frac{p\del_t}{M}$ and $1-p$ in Equation~\eqref{eq:insertions-deletions}, one can see that $\del_t$ tends to increase when $p\del_t/M < 1-p$ and that it tends to decrease when $p\del_t/M > 1-p$. So the equilibrium point is at $\delta_t \approx \frac{1-p}{p}M$. Fortunately for the analysis, we  show that a rehash occurs before  $\delta_t$ reaches this value, with exponentially high probability, so that at any step before the rehash $\delta_t$ is more likely to increase than to decrease. Let $t_h'$ denote the time of the next rehash. We can prove the following lemma.

\begin{restatable}{lemma}{lmbeforeeq}\label{lm:rehash before equilibrium}
For any positive $d>\frac1{2p-1}$, with exponentially high probability in $\free_{t_h}$, we have (i) $t_h'\leq t_h+\lceil d\free_{t_h}\rceil$, and (ii) at any time $t$  between $t_h$ and $t_h'$ we have $\frac{p\del_t}{M}\leq 1-p$.
\end{restatable}

\begin{proof}
The first statement is a direct consequence of Lemma~\ref{lm:random increase}, taking $s=t_h$, $t=\lceil d\free_{t_h}\rceil$ and $c$ such that  $\frac1d<c<2p-1$: with exponentially high probability in $\free_{t_h}$, there are more than $\used_{t_h}+c\lceil d\free_{t_h}\rceil > \used_{t_h}+\free_{t_h} = M$ used keys in the hashmap; they do not fit in a table of size $M$, hence a rehash already occurred and therefore $t_h'\leq t_h+\lceil d\free_{t_h}\rceil$.

For the second part we consider the processes $\tilde{\del}_t$ and $\tilde{\used}_t$ defined exactly as $\del_t$ and $\nu_t$ except that there are no border condition (rehash or empty table), so that they can continue indefinitely. Observe that 
$\tilde{\del}_t$ and $\del_t$  (resp. $\tilde{\used}_t$ and $\used_t$) coincide as long as no rehash is triggered and as long as the hashmap is not empty. We want to bound from above the probability that 
$\del_t > \frac{1-p}p M$ for some $t$ between $t_h$ and $t_h'$. By the union bound, this probability is upper bounded by the probability that $\tilde\del_t > \frac{1-p}p M$ for some $t$ between $t_h$ and $t_h+\lceil d\free_{t_h}\rceil$, plus some exponentially small probability for the cases where $\tilde{\del}_t$ and $\del_t$ cease to coincide. We can therefore focus on finding an upper bound for the quantity $P$, defined as follows:
\[
P:=\sum_{t=t_h}^{t_h+\lceil d\free_{t_h}\rceil} \Pr\left(\tilde{\del_t}=\left\lfloor\frac{1-p}{p}M\right\rfloor+1\text{ and }t\leq t'_h\right),
\]
using the fact that since $\tilde\del_t$ increases  by at most one at each step, $\tilde\del_t>\frac{1-p}{p}M$ implies that $\tilde\del_s=\left\lfloor\frac{1-p}{p}M\right\rfloor+1$ for some $s\leq t$. Clearly, $P$ is an upper bound for the probability that (ii) does not hold for $\tilde\delta$. 

 Moreover, as  we just rehashed at time $t_h$, we have $\tilde\delta_{t_h}=0$. Hence, the probabilities for $t\leq t_h+ \left\lfloor\frac{1-p}{p}M\right\rfloor$ in $P$ contribute to zero to the sum: not sufficiently many keys can have been inserted. Let $\kappa=\left\lfloor\frac{1-p}{p}M\right\rfloor+1$, our expression for $P$ simplifies to
\[
P=\sum_{t=t_h+\kappa}^{t_h+\lceil d\free_{t_h}\rceil} \Pr\left(\tilde{\del_t}=\kappa\text{ and }t\leq t'_h\right).
\]

Let $f\in(0,1)$ be a constant that only depends on $p$ to be chosen later on. We split $P$ in two parts $P_1$ and $P_2$ the following way, with the convention that the second sum is $0$ if $\lfloor fM/p\rfloor + 1 >\lceil d\free_{t_h}\rceil$:
\[
P \leq \underbrace{\sum_{t=t_h+\kappa}^{t_h+\lfloor fM/p\rfloor} \Pr\left(\tilde{\del_t}=\kappa\text{ and }t\leq t_h'\right)}_{P_1} + \underbrace{\sum_{t=t_h+\lfloor fM/p\rfloor+1}^{t_h+\lceil d\free_{t_h}\rceil} \Pr\left(\tilde{\del_t}=\kappa\text{ and }t\leq t_h'\right)}_{P_2}.
\]
A deletion is required for $\tilde\del_t$ to increases, hence
$\tilde\del_t \leq D_{t-t_h}$ where $D_k\sim\bin(k,1-p)$. For any $t$ such that $t_h+\kappa\leq t\leq t_h+ fM/p$ we have to perform at most $fM/p$ operations since time $t_h$, and we intuitively expect to have at most $\approx (1-p)fM/p$ deletions, which is smaller than $\kappa$ for sufficiently large $M$, as $f<1$. Formally, for sufficiently large $T$, $M$ is sufficiently large so that $\kappa = \lfloor \frac{1-p}{p}M \rfloor +1 >\frac{1-p}{p} \cdot \frac{1+f}2\cdot M$. Hence, for any $t$ such that $t_h+\kappa\leq t\leq t_h+ fM/p$,
\begin{align*}
\Pr\left(\tilde{\del_t}=\kappa\text{ and }t\leq t_h'\right)& \leq\Pr\left(\tilde\del_t =\kappa\right) \leq \Pr\left(\tilde\del_t > \frac{(1+f)(1-p)M}{2p} \right) \\
& \leq \Pr\left(\bin(t-t_h,1-p) >\frac{(1+f)(1-p)M}{2p} \right)\\
&\leq \Pr\left(\bin(\lfloor fM/p\rfloor,1-p) > \frac{(1+f)(1-p)M}{2p} \right).
\end{align*}
This quantity is exponentially small in $M$ by Hoeffding inequality (Proposition~\ref{prop:hoeffding-ineq}). Hence, $P_1$ is exponentially small in $M$, hence in $\del_{t_h}$, as a sum of a linear number of exponentially small quantities.

For $P_2$ we use the fact that if $\tilde\del_t+\tilde\used_t>M$ then a rehash was triggered and $t'_h>t$, which is therefore not possible. Since we are looking for a time $t$ where  $\tilde{\del_t}=\kappa$,  we weaken the condition to obtain
\[
P_2\leq \sum_{t=t_h+\lfloor fM/p\rfloor+1}^{t_h+\lceil d\free_{t_h}\rceil} \Pr\left(\tilde\used_t \leq M - \kappa \right)
\leq  \sum_{t=t_h+\lfloor fM/p\rfloor+1}^{t_h+\lceil d\free_{t_h}\rceil} 
\Pr\left(\tilde\used_t \leq \frac{2p-1}p\,M\right).
\]
Recall that $\tilde\used_{t_h} = \used_{t_h} > M/2$. By Lemma~\ref{lm:random increase}, after more than 
$fM/p$ additional operations, we have   $\tilde\used_t \geq M/2 + c\cdot fM/p$ with exponentially high probability in $M$, for any $c\in(0,2p-1)$. But if 
$\tilde\used_t \geq M/2 + c\cdot fM/p$ then
\[
\tilde\used_t - \frac{2p-1}p\,M \geq M\left(\frac12 + \frac{c\cdot f}p - \frac{2p-1}p\right). 
\]
As $c$ can be as close to $2p-1$ as needed and $f$ as close to $1$ as needed, we can choose them such that $cf-(2p-1)> -p/4$. Therefore, for such a choice of $c$ and $f$, we  have that $\Pr\left(\tilde\used_t \leq \frac{2p-1}p\,M\right)$ is exponentially small in $M$. Summing a linear number of an exponentially small probabilities yields  an exponentially small probability, hence $P_2$ is exponentially small in $M$, hence in $\free_{t_h}$, concluding the proof.
\end{proof}

Let $t_0=\lfloor\frac{1-p}{p}\free_{t_h}\rfloor < \free_{t_h}$. Between time $t_h+1$ and $t_h+t_0$, there are not enough operations to trigger a rehash or to empty the hashmap. 
As $\del_t$ increases by at most $1$ at each operation, we have
$\del_t\leq \frac{1-p}{p}\free_{t_h}$ for any $t$ such that $t_h\leq t \leq t_h+t_0$. Since the hashmap is more than half filled just after a rehash, we have $\free_{t_h} < \frac{\M}2$ and $\frac{p\del_{t}}M\leq \frac{1-p}2$, for $t_h\leq t\leq t_h+t_0$. Hence, until time $t_h+t_0$, we can bound from below the process $\del_t$ by a process that increases with probability $1-p$, decreases with probability $\frac12(1-p)$ and does not change otherwise. This observation yields the following result. 

\begin{lemma}\label{lm:t0}
Let  $t_0=\lfloor \frac{1-p}{p}\free_{t_h} \rfloor$. With exponentially high probability in $t_0$, we have $\del_{t_h+t_0}\geq \frac{(1-p)^2}{3p}\free_{t_h}$.
\end{lemma}

\begin{proof}
Using the observations made before stating the lemma, the process of Equation~\eqref{eq:insertions-deletions} is stochastically lower bounded by the  process $\del_t^-$  defined by $\del^-_{t_h}=\del_{t_h}=0$ and for all $t\in\{t_h,\ldots,t_h+t_0\}$:
\begin{equation}\label{eq:insertions-deletions-lower-bound}
	\del^-_{t+1} = 
	\begin{cases}
		\del^-_t-1 & \text{with probability }\frac12(1-p) \\
		\del^-_t & \text{with probability }\frac12(3p-1)\\
		\del^-_t + 1 & \text{with probability }1-p.
	\end{cases}
\end{equation}
The random variable $\del^-_{t+1}$ is the difference of two (dependent) random variables, both following a binomial law. Since the expected increase at each operation is $\frac12(1-p)$ and there are $\lfloor\frac{1-p}p\free_{t_h}\rfloor$ operations, we get the result by applying Hoeffding Inequality (Proposition~\ref{prop:hoeffding-ineq}): if the difference deviates from the expectation of more than $\alpha$ times the expectation, at least one of the two random variables deviates from its expectation of more than $\alpha/2$ times its expectation, which is exponentially unlikely.
\end{proof}

From Lemma~\ref{lm:rehash before equilibrium}, with exponentially high probability the probability that the number of deleted elements decreases at any given step is smaller than the probability it increases. So $\del_t$ remains linear in $\free_{t_h}$ after time $t_h+t_0$ and until a rehash occurs, as formalized in the following statement.

\begin{lemma}\label{lm:tau}
For any $d>\frac1{2p-1}$,
with exponentially high probability w.r.t. $\free_{t_h}$ there is a rehash at some time $t'_h \leq t_h + d\free_{t_h}$. Moreover, there exists some constant $\gamma\in(0,1)$ such that, just before the rehash, $\del_{t'_h-1} \geq \gamma\, \free_{t_h}$.
\end{lemma}

\begin{proof}
The evolution of $\del_t$ is lower bounded by a process that increases and decreases with the same probability, namely, $1-p$. So after $t'_h-t_0-t_h$ new operations made after time $t_0$, 
$\del_t-\del_{t_0+t_h}$ is stochastically lower bounded by the difference of two dependent random variable that both follow a binomial distribution of parameters $ t'_h-t_0-t_h$ and $1-p$. The expectation is $0$, so we conclude using Hoeffding Inequality (Proposition~\ref{prop:hoeffding-ineq})  and take $\gamma = \frac{(1-p)^2}{4p}$ to get some linear room and have an exponentially small probability that the property does not hold.
\end{proof}

We now have all the ingredients to prove our main theorem.

\begin{proof}[Proof of Theorem~\ref{thm:main proba hashmap}] with exponentially high probability in $T$, the hashmap reaches the capacity $\M$. When it is the first time it reaches this size, and as we only insert elements one at a time,  the hashmap contains $\M/2+1$ used cells and $\free_0 := \M/2 - 1$ free cells. By Lemma~\ref{lm:tau}, after some time another rehash occurs, with at least $\gamma\, \free_0$ deleted cells just before rehashing. So the newly created hashmap has capacity $M$ (it is exponentially unlikely to decrease) and contains at least $\gamma\, \free_0$ empty cells. 
Then we apply  Lemma~\ref{lm:tau} again, and there is another rehash into a hashmap of capacity $\M$, with at least $\gamma^2\,\free_0$ empty cells, etc. We continue while $\gamma^i\, \free_0\geq \sqrt{T}$, that is, a logarithmic (in $T$) number of times. At each rehash we have to insert all the used keys, and there are a linear number of them, so it globally costs $\Omega(T\log T)$ calls to the insertion function: $\Omega(\log T)$ rehashes each costing $\Theta(T)$ calls. It is very likely that such a sequence of rehashes will occur, since at some point we reach a capacity of $2M$, with exponentially high probability, by Lemma~\ref{lm:random increase}. When we sum the probabilities of error using the union bound, we sum a logarithmic number of error terms, which are all in $\bigOh(\exp{\small(-c\sqrt{T}\small)})$, so it is super-polynomially unlikely  that this scenario does not happen. This concludes the proof.
\end{proof}

\section{Hybrid tables in \Lua} 

\label{sec:hybrid-tables}

Recall that when keys are integers, \lua stores their values in the array part of the hybrid table~\cite{IFF05,LuaGems}. 
The array-part corresponds to a range $[1,2^a]$ of keys, or $\emptyset$ at the beginning.
To avoid wasting memory, \lua makes sure that more than half of the keys are being used
at any one time. When associating a value to a key into the table, if the key is an integer within the range of
the array-part, the value is simply inserted there. Otherwise, the pair key/value is inserted into the hashmap as explained in Section~\ref{subsec:hashmaps-algorithms}.
If the insertion into the hashmap provokes a rehash, 
we first compute the largest $a'\geq 0$ such that $[1,2^{a'}]$ contains
at least $2^{a'-1}+1$ keys from the hybrid table.\footnote{This is done in linear time by counting the number of integer keys between $2^{\ell-1}$ and $2^{\ell}$ for each $\ell$, for $1\leq 2^\ell\leq M$. } Then $A'=2^{a'}$ will be the new size
of the array-part, and the values of the keys within its range are placed there.   The rest of the
keys are placed in the hashmap, which has size $\M=2^\m$, where $\m$ is chosen
so that the total number of elements it contains is between $2^{\m-1}$ (strictly) and  $2^\m$. The insertions after the rehash are performed in the order of their position in the previous hashmap, each key going either to the hash-part or to the array-part if it is an integer smaller than or equal to $A'$. This process is exemplified in Figure~\ref{fig:example rebuild}.

\begin{figure}[t]
	\begin{tikzpicture}[xscale=1.05,yscale=.8]
		%% hashpart
		\draw[thick] (0,0) rectangle (4,1);
		\foreach \x in {0,1,...,3} 
		{
			\draw[thick] (\x,0) -- (\x,1);
		}
		\foreach \x/\y in {0/5\mapsto a, 1/9\mapsto s, 2/11\mapsto g, 3/7 \mapsto b} {
			\node at (\x+0.5, .5) {\scriptsize $\y$};
		}
		\node at (-1,.5) {\small hashpart};
		
		%% arraypart
		\draw[thick] (6.5,0) rectangle (10.5,1);
		\foreach \x in {1,2,3,4} 
		{
			\draw[thick] (\x+6.5,0) -- (\x+6.5,1);
			\node at (\x+6, 1.2) {\scriptsize \x};
		}
		\foreach \x/\y in {1/p, 2/c, 4/c} {
			\node at (\x+6, .5) {\small $\y$};
		}
		\node at (5.5,.5) {\small arraypart};

		\draw[blue,->,very thick](4.5, -.2) -- node[right,align=left] {\small adding $12\mapsto f$\\\small triggering a rehash}(4.5, -1.8);

		%% hashpart 2
		\draw[thick] (0,-2) rectangle (4,-3);
		\foreach \x in {0,1,...,3} 
		{
			\draw[thick] (\x,-2) -- (\x,-3);
		}
		\foreach \x/\y in {0/12\mapsto f, 1/9\mapsto s,2/, 3/11 \mapsto g} {
			\node at (\x+0.5, -2.5) {\scriptsize $\y$};
		}
		\node at (-1,-2.5) {\small hashpart};

		%% arraypart 2
		\draw[thick] (0,-3.5) rectangle (8,-4.5);
		\foreach \x in {1,...,8} 
		{
			\draw[thick] (\x,-3.5) -- (\x,-4.5);
			\node at (\x-.5, -3.3) {\scriptsize\x};
		}
		\foreach \x/\y in {1/p, 2/c, 4/c, 5/a, 7/b} {
			\node at (\x-.5, -4) {\small $\y$};
		}

		\node at (-1,-4) {\small arraypart};

	\end{tikzpicture}

    \caption{
    \label{fig:example rebuild}
    Example of a \Lua table, with its hash-part and array-part. On the top, we have the initial state of the table, with the hash-part full. This is a valid configuration, it can be obtained in \Lua from an empty table if we insert the keys in the order $1,2,4,11,9,7,5$. Inserting $12\mapsto f$ induces a rehash. Then the size of the array-part is recomputed, choosing the largest interval of keys $[1,2^{a^\prime}]$ so that it will be more than half-full. The situation after the rehash can be seen at the bottom.%, thus simulating a deletion.
    }
\end{figure}

\subsection{Settings for the analysis}
\label{subsec:settings-analysis-array}
In all the following analysis of \lua hybrid tables, we consider that only insertions of pairs key/value are performed. This setting is  sufficient to exhibit some unfortunate behavior in natural models, and it can only become worse if we also allow deletions. Rather than being interested
in what happens between two rehashes, as in the previous section, we study what happens when a rehash occurs.

We consider a sequence of $n$ insertions ${\boldsymbol{y}}=(y_1,\ldots,y_n)$ of integers (keys)
into an initially empty \lua table.
We write $t_0,t_1,t_2,\ldots,t_\ell$ for the sequence of rehash times, setting $t_0=0$,
letting $t_i$ be the time of the \nth{i} rehash. Let $\ell=\ell(\boldsymbol{y})$ be the total number of rehashes.
More precisely, the insertion of $y_{t_i}$ induces the \nth{i} rehash.
Denote by $\beta_i$
the size of the hashmap after the \nth{i} rehash. 

In the process, the cost introduced by the insertions of elements in the array-part
is small: as we only consider insertions in this section, the size of the array-part can only increase, and the amortized cost of an insertion in such a dynamic array is $\bigOh(1)$, so the overall cost is $\bigOh(n)$. Hence, the cost of interest for us is the one induced by the rehashes, and we denote by $C$ the \emph{cost} defined by
$C:=\sum \beta_i$.
This exactly estimates the over-cost induced by rehashes, and is our main parameter of study in this section. 
It is an accurate estimation of the total number of calls to the \lua insertion function, up to a multiplicative constant.

\subsection{Rehashing into the array-part}%: a new kind of deletion}
\label{subsec:rehash-arrays}
In this section, we show a simple example of how the hybrid mechanism
may lead to super-linear costs $C=\Omega(n \log n)$. Moreover, we prove that $\bigOh(n\log n)$ 
is the worst case possible for a sequence $n$ insertions into an
initially empty \lua table.

\begin{example} 
\label{example:hybrid-1}
Consider inserting
$(-(2^k-1),-(2^k-2),\ldots,0,1,2,\ldots,2^k)$ into an empty \lua table.
We claim that this induces exactly $k$ rehashes, in which we systematically 
reinsert $2^k$ entries into a renewed hash-part of size $2^k$.

This is proved as follows. Non-positive integers go  into the hash-part of the table. Thus $(-(2^k-1),-(2^k-2),\ldots,0$ go into the hashmap, which is then full and of size $\M=2^k$. Then inserting $1$, induces a rehash. However, the key $1$ goes immediately to the array-part. Continuing, we rehash and double the size of the array-part when inserting $1,2,3,5,9,\ldots$, in short for $1$ and $(2^i+1)_{i=0}^{i=k-1}$. The remaining positive keys go directly into a free spot of the array-part and do not induce a rehash.
\end{example}

It is therefore possible for a sequence
of $n=2^{k+1}$ integers to yield a cost $C = \Omega(n\log n)$. This is the worst  possible case:
\begin{restatable}{proposition}{propworstcase}\label{worst-case-permutations-2}
When performing $n$ insertions into an initially empty \lua table, the insertion function 
is called $\bigOh(n \log n)$ times in the worst-case.
\end{restatable}

We recall that we denote by $t_0,t_1,t_2,\ldots,t_\ell$ the sequence of rehash times, setting $t_0=0$,
letting $t_i$ be the time of the \nth{i} rehash, and $\ell=\ell(\boldsymbol{y})$ be the total number of rehashes.
Thus, the insertion of $y_{t_i}$ induces the \nth{i} rehash.
We denote by $\beta_i$
the size of the hashmap after the \nth{i} rehash, and let $\alpha_j$ be the
number of elements in the array-part right before performing the \nth{j} rehash. We also denote by $\alpha_{\ell+1}$ the final number of elements in the array-part, as there
is no \nth{(\ell+1)} rehash, and we set $t_{\ell+1}:=\alpha_{\ell+1}+\beta_\ell +1$.
\begin{remark} \label{remark:rehash-times2}
With this notation, the rehash times satisfy $\alpha_{j+1}+\beta_j+1=t_{j+1}$, for all $j$ such that $0\leq j\leq \ell$.
\end{remark}
\begin{remark}
    Note that when considering only insertions, at each rehash, either the size (capacity) of the array-part increases
    or of the hash-part increases. Indeed, when an element triggers
    a rehash, it cannot go into the array-part (else there would be no rehash) and the hash-part must be full. If the hash-part does not increase in size, then the array-part has to.
\end{remark}

\begin{proof}[Proof of Proposition~\ref{worst-case-permutations-2}]
We want to show that the cost $C:=\sum_{i=0}^\ell \beta_i$ is $\bigOh(n \log n)$. For this, we distinguish the $\beta_i$'s that correspond to an increase of the array-part size from the other ones.
 
When the array-part increases in size, at the very least it doubles in size. Hence, the total number of rehashes in which the array-part increases is $\bigOh(\log n)$. The hash-part does not increase when the array-part does, as elements are added one at a time. Since $\beta_i \leq 2n$ for all $i$,
the overall contribution of these $\beta_i$'s to $C$ is at most 
$\bigOh(n\log n)$.

We now show that the remaining rehashes only contribute a total of $\bigOh(n)$.  Consider the values of $i$ for which the hash-part doubles (with $1\leq i \leq \ell$), that is, such that $\beta_{i}=2\beta_{i-1}$. Then $\beta_i=2(\beta_{i}-\beta_{i-1})$ we have $\sum_{i : \beta_{i}>\beta_{i-1}} \beta_i = 2\sum_{j : \beta_{j+1}>\beta_j} (\beta_{j+1}-\beta_{j}) $. Observe that
$
\sum_{j : \beta_{j+1}>\beta_j} (\beta_{j+1}-\beta_j) = \sum_{j=0}^{\ell-1} (\beta_{j+1}-\beta_j) - \sum_{j : \beta_{j+1}\leq \beta_j} (\beta_{j+1}-\beta_j)\,,
$
and $\sum_{j=0}^{\ell-1} (\beta_{j+1}-\beta_j)=\beta_\ell \leq 2n$.

To conclude, we remark that  $\sum_{j : \beta_{j+1}\leq \beta_j} |\beta_{j+1}-\beta_j| \leq n $. This is because the keys leaving the hash-part go into the array-part once and never come back.  We make this connection formal. When rehashing we have $\alpha_{j+1}+\beta_j+1=t_{j+1}$ (see Remark~\ref{remark:rehash-times2}), thus $\beta_j-\beta_{j+1} = \alpha_{j+2}-\alpha_{j+1} + (t_{j+1}-t_{j+2}) < \alpha_{j+2}-\alpha_{j+1}\,.$ Then, as $\alpha_j$ is non-decreasing,  $|\beta_{j+1}-\beta_j|\leq \alpha_{j+2}-\alpha_{j+1}$ when $\beta_{j+1}\leq \beta_j$. Again using that $\alpha_j$ is non-decreasing we obtain \[\sum_{j : \beta_{j+1}\leq \beta_j}|\beta_{j+1}-\beta_j|  \leq \sum_{j : \beta_{j+1}\leq \beta_j}(\alpha_{j+2}-\alpha_{j+1}) \leq  \sum_{j=0}^{\ell-1} (\alpha_{j+2}-\alpha_{j+1}) = \alpha_{\ell+1}-\alpha_1 \leq n \,. \]

Therefore $\sum_{i : \beta_{i}>\beta_{i-1}} \beta_i \leq 6n$, concluding the proof.
\end{proof}

\begin{remark} \label{rem:rehashes-hash-increases}
    The proof of \autoref{worst-case-permutations-2} shows that the rehashes in which the hash-part increases have a global contribution of only $\bigOh(n)$, while the others contribute for at most $\bigOh(n \cdot \log (1+m))$, where $m$ is the maximum size of the array part.
\end{remark}

\begin{remark}
If some of the $n$ inserted keys are not positive integers, we can refine the result.
Let $n^\prime$ be the total number of positive integer keys. Then the worst case for $C$ is $\bigOh(n+n\log (1+n^\prime))$, as long as only insertions are performed.
\end{remark}

\subsection{Inserting permutations in \Lua Tables}
\label{subsec:inserting-permutations}

In Example~\ref{example:hybrid-1} we showed a sequence of $n$ insertions of integers yielding a cost $C=\Theta(n\log n)$. Some of the keys were negative integers,
so they could only ever be stored in the hash-part. In this subsection, we show that this worst case is still attainable on a very natural setting involving only positive integers: the sequence ${\boldsymbol{y}}$ is a permutation of $[n]:=\{1,\ldots,n\}$. This setting, a priori, gives the array-part the best possible chance of being exploited, while not repeating keys. We present
both the worst-case (this subsection) and the case of a random permutation (Subsection~\ref{subsec:analysis-insertion-random-permutations}). Though described in terms of permutations, these settings apply whenever the keys are consecutive integers that do not necessarily appear in increasing order. It is the case for instance during the marking of the transversal of a graph of vertex set $[n]$.

\begin{restatable}{proposition}{permutationsworst}\label{worst-case-permutations-1}
Inserting $n$ elements given by the order of a permutation $\pi$ of  $[n]$  requires $\Omega(n \log n)$ calls to the insertion function of \lua in worst-case.
\end{restatable}
\begin{proof}
Consider first $n=3\times 2^k$ for some $k>0$. Define the permutation
\[
 \pi = \left(\begin{smallmatrix} 
1 & 2 & \cdots & 2^k & 2^k+1 &2^k+2 &\ldots & 3\cdot 2^k \\
2\cdot 2^k + 1 & 2\cdot 2^k+2 & \ldots & 3\cdot 2^k & 1 & 2& \ldots & 2\cdot 2^k
\end{smallmatrix}\right) \,.
\]

Notice that $2\cdot 2^k+1,\ldots 3\cdot 2^k$ cannot be on the array-part, unless the keys of the array-part contain the range $[1,\ldots,4\cdot 2^k]$. For this to happen more than half of the elements of that range must have been inserted, and this only happens after a time $t$ such that $\pi(t)=2^k$, so at time $t=2^{k+1}$.
Thus, after inserting $\pi(1),\ldots,\pi(2^k)$, the hash-part has size $\M=2^k$ and it is full. Inserting $\pi(2^k+1)=1$ induces a rehash, but $1$ goes into the array-part. Similarly, for $\pi(2^k+2)$, and more generally the same is true for the insertion of $\pi(2^k+2^j+1)$ as long as $j < k$. The rest are inserted directly into the array-part. For example, the insertion of $4=\pi(2^k+4)$ does not produce a rehash as $4$ goes into the array-part of range $[1,4]$. Each of the $k$ rehashes we have just described  costs at least $2^k$ because all the elements $\pi(1),\ldots,\pi(2^k)$ are still in the hash-part.
Hence, we obtain a cost that is $\Omega(2^k \times k) = \Omega(n  \log n)$.

For general $n$, we  pick the largest integer $k$ such that $m=3\times 2^k \leq n$, and complete the permutation $\pi$ above by setting $\pi(i)=i$ for $i\geq m+1$.
\end{proof}

\subsection{Average case for insertion of permutations}
\label{subsec:analysis-insertion-random-permutations}
The average case for permutations is almost linear, but not because of the wrong reasons. Our main result of this section, Theorem~\ref{thm:average-case-permutations2}, tells us that, essentially for any {\em super-linear} function $h(n)$,  the cost $C$ of a random permutation is $\bigOh(h(n))$ with probability
tending to $1$. There is a caveat with the choice of $n$: the theorem does not
apply if $n\to\infty$ approximates powers of two from above. Theorem~\ref{thm:average-case-permutations2} enforces this restriction 
by introducing subsets $\mathbb{N}_b:=\bigcup_{j=0}^\infty \{n : 2^j b < n \leq 2^{j+1}\}$, parameterized by $b\in (1,2)$.
This restriction is technical, but not restrictive for our purpose as it still applies to most natural choices of $n$, e.g., powers of two $n=2^k$. We suspect that such a restriction might be necessary; we are not sure if the conclusions apply to the sub-sequence $n=2^k+1$, and the arguments become too involved for the main topic of this article to justify an in-depth study here.

Even though inserting the keys of $[n]$ in a permuted order is essentially done almost linear time, we prove that the array-part is not really exploited as it should. For a uniform random permutation, with high probability, 
the corresponding \lua table does not have an array-part until the very end. This can be seen in Lemma~\ref{lemma:half-set-permutation-2nd-moment} below. Thus, in this scenario, we do not really take advantage of the hybrid data structure.

\begin{theorem} \label{thm:average-case-permutations2}
Let $g\colon \mathbb{N}\to \mathbb{R}_{>0}$ be an arbitrary function
satisfying $g(n)\to\infty$, with $g(n)=o({n})$ as $n\to\infty$. Fix an arbitrary $b\in (1,2)$ and let $n\to\infty$ over $\mathbb{N}_b$. Then, with probability tending to $1$ as $n\to\infty$, the number of calls to the insertion function, starting from an empty table, to build the \lua table for a uniform random permutation of $[n]$ is $\bigOh(n g(n))$.
\end{theorem}

Fixing $S\subseteq [n]$, let $S_t = S_t(\pi):= \{\pi(k) : k \leq t\}\cap S$, in other words, the set of keys from $S$ revealed up to time $t$. 
The proof of Theorem~\ref{thm:average-case-permutations2} is based on a key lemma and its corollary.
Consider a time $t \leq c n$ with $c<1/2$. We show in Lemma~\ref{lemma:half-set-permutation-2nd-moment} that, if $|S|$ is large, then time $t$ will not be enough
to have revealed half of the elements of $S$, i.e., $|S_t|>|S|/2$ is unlikely. Lemma~\ref{lemma:half-set-permutation-2nd-moment} and its corollary, Corollary~\ref{cor:half-set-permutation-2nd-moment}, apply to all $n\in \mathbb{N}$. In particular, the latter implies that array-part is not really used at any time $t\leq c n$.

\begin{lemma}
\label{lemma:half-set-permutation-2nd-moment}
Let $\theta = \tfrac{t}{n} < \tfrac{1}{2}$ and $s=|S|$, then
\begin{equation*}
\Pr\left(|S_t| > s/2\right) \leq \frac{\theta}{(\tfrac{1}{2}-\theta)^2} s^{-1}\,.
\end{equation*}
\end{lemma}
\begin{corollary}
    \label{cor:half-set-permutation-2nd-moment}
    Fix $c<1/2$ and let $g(n)\to\infty$. With probability tending to one, none of the sets $A_j=[1\ldots 2^j]$, for $2^j\geq g(n)$, is half-full at time $t \leq c n$.
\end{corollary}
\begin{proof}
    The probability that $A_j$ is more than half-full is at most $c \left(\tfrac{1}{2}-c\right)^{-2} 2^{-j}$ by Lemma~\ref{lemma:half-set-permutation-2nd-moment}. Thus, by the union bound, the probability that at least one of the $A_j$'s is more than half-full is bounded by
    $
    \sum_{j\geq \log_2 g(n)} c \left(\tfrac{1}{2}-c\right)^{-2} 2^{-j} = c \left(\tfrac{1}{2}-c\right)^{-2} 2^{1-\lceil \log_2 g(n)\rceil}
    $. As $g(n)\to\infty$, the bound tends to $0$.
\end{proof}

\begin{remark}
    One can directly obtain $\Pr\left(|S_k|=i\right)= \binom{|S|}{i} \cdot {\binom{n-|S|}{t-i}}/{\binom{n}{t}} $ using a simple counting argument. This is a hypergeometric distribution.
    More precise tail inequalities are known for this distribution \cite{CHVATAL1979285}. For our purposes, the bound in Lemma~\ref{lemma:half-set-permutation-2nd-moment} is enough, as we only consider  $S=[1,2^j]$ for $j\geq 1$. 
\end{remark}

\subsubsection{Proof of Lemma~\ref{lemma:half-set-permutation-2nd-moment}}
We first recall the notation we use for the proof.  Given a time $t\leq n$, let $\theta=\tfrac{t}{n}$. Fixing $S\subseteq [n]$, for a permutation $\pi$ of $\{1,\ldots,n\}$, we write $S_t = S_t(\pi):= \{\pi(k) : k \leq t\}\cap S$, $s=|S|$ and $s_t=s_t(\pi):=|S_t(\pi)|$

We start with the proof of Lemma~\ref{lemma:half-set-permutation-2nd-moment}. For the proof, we apply Chebyshev's inequality. Write
\[
s_t(\pi) = \sum_{i\in S} X_i^{\langle t\rangle}(\pi)
\]
where $X_i^{\langle t\rangle}(\pi)$ is $1$ if $i$ belongs to
$\{\pi(1),\ldots,\pi(t)\}$ and $0$ otherwise. For simplicity, we write $X_i=X_i^{\langle t\rangle}$ omitting the time $t$, as it will be fixed.

\begin{lemma}
    \label{lemma:expected-elements-subsets}
    For $A\subset \{1,\ldots,n\}$, let  $X_A = \prod_{i\in A} X_i$. Then
    $\mathbb{E}[X_A] = \binom{n-|A|}{t-|A|} / \binom{n}{t} $. Moreover, the inequality $\mathbb{E}[X_A] \leq (\tfrac{t}{n})^{|A|} $ holds.
\end{lemma}
\begin{proof}
    The variable $X_A$ is $1$ if $S\subseteq \{\pi(1),\ldots,\pi(t)\}$ and $0$ otherwise. Disregarding the order of $\pi(1),\ldots \pi(t)$, as it does not affect the sets, the number of sets of size $t$ having this property is $\binom{n-|A|}{t-|A|}$ as we must enforce that the elements of $A$ are present. Hence, the first equality.

    For the inequality, we note that if $|A|\leq t$,
    $
    \binom{n-|A|}{t-|A|} / \binom{n}{t} = \frac{t \cdots (t-|A|+1)}{n\cdots (n-|A|+1)}\,,
    $
    and $\tfrac{t-i}{n-i}\leq \tfrac{t}{n}$ for $i < n $. If $|A| > t$ the inequality is trivial.
\end{proof}

By symmetry and linearity of the expectation,  $\mathbb{E}[s_t] = s \times \mathbb{E}[X_1] = s \times \binom{n-1}{t-1} / \binom{n}{t} = s \times \theta$, yielding the expected value of $s_t$.

For the second moment of $s_t$, we observe that classically
\[\mathbb{E}[s_t^2] = \sum_{i\in S} \mathbb{E}[X_i^2] + 2  \sum_{i,j\in S : i<j} \mathbb{E}[X_i X_j] =  \mathbb{E}[|S_t|]+ 2  \sum_{i,j\in S : i<j}  \mathbb{E}[X_i X_j]\,. \]
By Lemma~\ref{lemma:expected-elements-subsets}, $\mathbb{E}[X_i X_j] = \binom{n-2}{t-2} / \binom{n}{t} \leq \theta^2$ where $\theta=t/n$. Thus, by symmetry again and ${\tt Var}(s_t)=\mathbb{E}[s_t^2] - \mathbb{E}[s_t]^2 $, we obtain
\[\mathbb{E}[s_t^2] \leq s\theta + s^2 \theta^2 \,,\qquad {\tt Var}(s_t)  \leq  s\theta\,. \]

Hence, by Chebyshev's inequality, we have proved Lemma~\ref{lemma:half-set-permutation-2nd-moment}.

\subsubsection{Proof of Theorem~\ref{thm:average-case-permutations2}}
We are now ready to prove the theorem. Recall that when we say at time $t$, we mean after the $t$-th insertion. The notation $a(n)\ll b(n)$ and $b(n)\gg a(n)$ means $a(n)<b(n)$ for all large enough $n$.

The following simple remark will be important in the proof. 
\begin{remark}
    \label{rem:rehashes-hash-increases-2}
    Suppose we only insert elements into an initially empty hybrid table. If, after the $t_0$-th insertion, the hash-part has capacity $M=2^h$ and the array-part $A=2^a$, then $t_0 \leq M+A$ and
    the next rehash is triggered by an insertion between $t=M+1$ and $t=M+A+1$ (inclusive).
\end{remark}
\begin{proof}[Proof of Theorem~\ref{thm:average-case-permutations2}]

Write $2^k b < n \leq 2^{k+1}$, where $b\in (1,2)$.

The proof consists of three steps:
\begin{enumerate}
    \item 
    By Corollary~\ref{cor:half-set-permutation-2nd-moment}, at time $t= 2g(n) + 2^{k-1} $ the array-part has size less than $g(n)$, with probability tending to $1$. Indeed, if $b^\prime\in(1,b)$, since $g(n)=o(n)$ we have  $t= 2g(n) + 2^{k-1} \ll n/(2b^\prime)$, and we apply Corollary~\ref{cor:half-set-permutation-2nd-moment} with $c=1/(2b^\prime)$. Thus, the hash-part must have size $\M=2^{k}$, because $2^{k-1} + g(n)\ll t$ and so $\M\leq 2^{k-1}$ is impossible by Remark~\ref{rem:rehashes-hash-increases-2}, while $\M=2^{k+1}$ would be absurd as $\M / 2 = 2^k \gg t$, by our hypothesis $g(n)=o(n)$, and the hash-part would not be half-full.
    \item By Remark~\ref{rem:rehashes-hash-increases-2}, the next rehash occurs at $t > 2^k$. Then the array-part takes up the whole range $[1,2^{k+1}]$. Thus, the total number of rehashes in which the array-part can increase is at most $2+\log_2 g(n)=\bigOh(g(n))$.
    \item Then we conclude by Remark~\ref{rem:rehashes-hash-increases}: the rehashes in which the array-part increases contribute at most $\bigOh(n\,g(n))$ to the cost $C$, while the rest of the rehashes contribute only $\bigOh(n)$. \qedhere
\end{enumerate}

 \end{proof}

\section{Conclusions and final remarks}
\label{sec:conclusions}

The only data-structuring mechanism in \lua are tables, so it is of the upmost importance that they are extremely efficient in time and space usage. \Lua hybrid tables work very well in many practical use-cases, for example, to create an array of $n$ elements filling $A[1]$, \ldots, $A[n]$ sequentially, or to build a dictionary in which we alternate insertions and searches but have no deletions, or when we fill a table and then process and remove one by one its elements. But there are some situations which may also arise in practice rather naturally where there are noticeable inefficiencies or suboptimal performance of the \lua hybrid tables, as our theoretical analysis has revealed. Fig.\ref{fig:experiments-time} illustrates that this unwanted behavior shows in practice on our simulations.

\begin{figure}[ht]
    \centering
    \begin{tikzpicture}[scale=0.6]
  \begin{axis}[
%   xmode=log,
%   log basis x={10},
%     ymode=log,
%   log basis x={10},
   xlabel={number of operations $T$},
    ylabel={${\tt total\ time\ (seconds)}$},every y tick scale label/.append style={xshift=-1.6em},every x tick scale label/.append style={xshift=2.8em, yshift=1.4em},% xmode=log,   log basis x={2}
    ]
    
        \addplot  [blue]  coordinates {
(100000, 0.062974)
(200000, 0.142318)
(300000, 0.189110)
(400000, 0.337570)
(500000, 0.386275)
(600000, 0.441415)
(700000, 0.747508)
(800000, 0.799991)
(900000, 0.853294)
(1000000, 0.907695)
(1100000, 0.963363)
(1200000, 1.026785)
(1300000, 1.187914)
(1400000, 1.740708)
(1500000, 1.794744)
(1600000, 1.849799)
(1700000, 1.905969)
(1800000, 1.962294)
(1900000, 2.019209)
(2000000, 2.077247)
(2100000, 2.136185)
(2200000, 2.195670)
(2300000, 2.256305)
(2400000, 2.338443)
(2500000, 2.400596)
(2600000, 2.691898)
(2700000, 3.920978)
(2800000, 3.975217)
(2900000, 4.030007)
(3000000, 4.085194)
(3100000, 4.140767)
(3200000, 4.196560)
(3300000, 4.252615)
(3400000, 4.309260)
(3500000, 4.366320)
(3600000, 4.423437)
(3700000, 4.481256)
(3800000, 4.539618)
(3900000, 4.598445)
(4000000, 4.657263)
(4100000, 4.716530)
(4200000, 4.776273)
(4300000, 4.836781)
(4400000, 4.897618)
(4500000, 4.958849)
(4600000, 5.020686)
(4700000, 5.108101)
(4800000, 5.171053)
(4900000, 5.233760)
(5000000, 5.296874)
(5100000, 5.827059)
(5200000, 5.891031)
(5300000, 8.531626)
(5400000, 8.586637)
(5500000, 8.641549)
(5600000, 8.696772)
(5700000, 8.752620)
(5800000, 8.808608)
(5900000, 8.864608)
(6000000, 8.920687)
(6100000, 8.977162)
(6200000, 9.034081)
(6300000, 9.091124)
(6400000, 9.147789)
(6500000, 9.205484)
(6600000, 9.263010)
(6700000, 9.320362)
(6800000, 9.377855)
(6900000, 9.435415)
(7000000, 9.492939)
(7100000, 9.550550)
(7200000, 9.608197)
(7300000, 9.665909)
(7400000, 9.723404)
(7500000, 9.781381)
(7600000, 9.839653)
(7700000, 9.898606)
(7800000, 9.956998)
(7900000, 10.015817)
(8000000, 10.074632)
(8100000, 10.133639)
(8200000, 10.192783)
(8300000, 10.252104)

%  p = 0.9
     };

        \addplot  [blue,dashed]  coordinates {
(100000, 0.039734)
(200000, 0.086341)
(300000, 0.126802)
(400000, 0.193323)
(500000, 0.236783)
(600000, 0.288747)
(700000, 0.390563)
(800000, 0.440078)
(900000, 0.490559)
(1000000, 0.541962)
(1100000, 0.594094)
(1200000, 0.658843)
(1300000, 0.827137)
(1400000, 0.878910)
(1500000, 0.931342)
(1600000, 0.984409)
(1700000, 1.038195)
(1800000, 1.092628)
(1900000, 1.147890)
(2000000, 1.203830)
(2100000, 1.260585)
(2200000, 1.318187)
(2300000, 1.376616)
(2400000, 1.448492)
(2500000, 1.508509)
(2600000, 1.792853)
(2700000, 1.845606)
(2800000, 1.898647)
(2900000, 1.952052)
(3000000, 2.005897)
(3100000, 2.059978)
(3200000, 2.114371)
(3300000, 2.169176)
(3400000, 2.224335)
(3500000, 2.279853)
(3600000, 2.335679)
(3700000, 2.391904)
(3800000, 2.448468)
(3900000, 2.505403)
(4000000, 2.562718)
(4100000, 2.620461)
(4200000, 2.678574)
(4300000, 2.737015)
(4400000, 2.795843)
(4500000, 2.855161)
(4600000, 2.914810)
(4700000, 2.997370)
(4800000, 3.057797)
(4900000, 3.118621)
(5000000, 3.179437)
(5100000, 3.697339)
(5200000, 3.749859)
(5300000, 3.802529)
(5400000, 3.855368)
(5500000, 3.908362)
(5600000, 3.961501)
(5700000, 4.014814)
(5800000, 4.068293)
(5900000, 4.121869)
(6000000, 4.175579)
(6100000, 4.229425)
(6200000, 4.283397)
(6300000, 4.337456)
(6400000, 4.391628)
(6500000, 4.445922)
(6600000, 4.500264)
(6700000, 4.554745)
(6800000, 4.609612)
(6900000, 4.664484)
(7000000, 4.719385)
(7100000, 4.774273)
(7200000, 4.829312)
(7300000, 4.884382)
(7400000, 4.939512)
(7500000, 4.994753)
(7600000, 5.049964)
(7700000, 5.105017)
(7800000, 5.160255)
(7900000, 5.215294)
(8000000, 5.270476)
(8100000, 5.325560)
(8200000, 5.380704)
(8300000, 5.435880)
% p=0.9, with "early rehash" to a hastable at most 80% full
};

             \addplot  [red!70!white]  coordinates {
(100000, 0.051785)
(200000, 0.111401)
(300000, 0.211995)
(400000, 0.253456)
(500000, 0.312290)
(600000, 0.499629)
(700000, 0.547309)
(800000, 0.593067)
(900000, 0.640544)
(1000000, 0.729006)
(1100000, 1.123743)
(1200000, 1.172641)
(1300000, 1.222305)
(1400000, 1.280520)
(1500000, 1.331907)
(1600000, 1.383518)
(1700000, 1.435972)
(1800000, 1.489104)
(1900000, 1.632236)
(2000000, 1.686417)
(2100000, 2.468658)
(2200000, 2.653856)
(2300000, 2.704869)
(2400000, 2.756446)
(2500000, 2.808578)
(2600000, 2.861103)
(2700000, 2.914069)
(2800000, 2.986120)
(2900000, 3.040546)
(3000000, 3.095438)
(3100000, 3.150563)
(3200000, 3.206139)
(3300000, 3.262699)
(3400000, 3.319591)
(3500000, 3.377583)
(3600000, 3.435946)
(3700000, 3.692681)
(3800000, 3.750723)
(3900000, 3.809358)
(4000000, 3.868688)
(4100000, 4.136277)
(4200000, 5.562277)
(4300000, 6.129023)
(4400000, 6.181224)
(4500000, 6.233648)
(4600000, 6.286186)
(4700000, 6.338959)
(4800000, 6.391711)
(4900000, 6.444780)
(5000000, 6.498154)
(5100000, 6.551745)
(5200000, 6.605329)
(5300000, 6.659265)
(5400000, 6.713333)
(5500000, 6.767760)
(5600000, 6.846858)
(5700000, 6.902120)
(5800000, 6.957859)
(5900000, 7.013776)
(6000000, 7.070054)
(6100000, 7.126778)
(6200000, 7.183415)
(6300000, 7.240264)
(6400000, 7.297589)
(6500000, 7.354932)
(6600000, 7.412562)
(6700000, 7.470267)
(6800000, 7.528494)
(6900000, 7.586742)
(7000000, 7.645175)
(7100000, 7.703846)
(7200000, 7.762939)
(7300000, 8.230969)
(7400000, 8.289634)
(7500000, 8.348605)
(7600000, 8.407774)
(7700000, 8.467237)
(7800000, 8.526916)
(7900000, 8.587014)
(8000000, 8.646993)
(8100000, 9.141632)
(8200000, 9.202003)
(8300000, 9.686099)
% p = 0.75
};

             \addplot  [red!70!white,dashed]  coordinates {
(100000, 0.034741)
(200000, 0.072709)
(300000, 0.118555)
(400000, 0.158864)
(500000, 0.222191)
(600000, 0.263816)
(700000, 0.311491)
(800000, 0.356544)
(900000, 0.452132)
(1000000, 0.499755)
(1100000, 0.548173)
(1200000, 0.596981)
(1300000, 0.646701)
(1400000, 0.707359)
(1500000, 0.758425)
(1600000, 0.809910)
(1700000, 0.861638)
(1800000, 1.020947)
(1900000, 1.071174)
(2000000, 1.121833)
(2100000, 1.172944)
(2200000, 1.224443)
(2300000, 1.276432)
(2400000, 1.328886)
(2500000, 1.381729)
(2600000, 1.434843)
(2700000, 1.488354)
(2800000, 1.556831)
(2900000, 1.611361)
(3000000, 1.666316)
(3100000, 1.721844)
(3200000, 1.777886)
(3300000, 1.834454)
(3400000, 1.891575)
(3500000, 1.949282)
(3600000, 2.221103)
(3700000, 2.272199)
(3800000, 2.323458)
(3900000, 2.375030)
(4000000, 2.426767)
(4100000, 2.478827)
(4200000, 2.531141)
(4300000, 2.583579)
(4400000, 2.636303)
(4500000, 2.689196)
(4600000, 2.742271)
(4700000, 2.795591)
(4800000, 2.849168)
(4900000, 2.903011)
(5000000, 2.957073)
(5100000, 3.011333)
(5200000, 3.065807)
(5300000, 3.120484)
(5400000, 3.175272)
(5500000, 3.230339)
(5600000, 3.309795)
(5700000, 3.365614)
(5800000, 3.421539)
(5900000, 3.477645)
(6000000, 3.533945)
(6100000, 3.590522)
(6200000, 3.647404)
(6300000, 3.704789)
(6400000, 3.762259)
(6500000, 3.820094)
(6600000, 3.878113)
(6700000, 3.936250)
(6800000, 3.994887)
(6900000, 4.053705)
(7000000, 4.112646)
(7100000, 4.607434)
(7200000, 4.658836)
(7300000, 4.710335)
(7400000, 4.762010)
(7500000, 4.813857)
(7600000, 4.865918)
(7700000, 4.918065)
(7800000, 4.970463)
(7900000, 5.023340)
(8000000, 5.075890)
(8100000, 5.128359)
(8200000, 5.180882)
(8300000, 5.233634)
% p = 0.75 with "early rehash" to a hashtable at most 80% full
};

\end{axis}
\end{tikzpicture} \qquad
             \begin{tikzpicture}[scale=0.6]
  \begin{axis}[
%   xmode=log,
%   log basis x={10},
%     ymode=log,
%   log basis x={10},
  ymin=.5,
  xlabel={number of operations (log scale)},
    ylabel={amortized time per operation ($\mu{\tt s}/{\tt op}$)},every y tick scale label/.append style={xshift=-1.6em},every x tick scale label/.append style={xshift=2.8em, yshift=1.4em}, 
       xmode=log,
   log basis x={2},
]
    % values are multiplied by 1587956.9346079333 so that we get a nice range on y
        \addplot  [blue]  coordinates {

(100000.0,0.6297400000000001)
(200000.0,0.7115900000000001)
(300000.0,0.6303666666666667)
(400000.0,0.843925)
(500000.0,0.77255)
(600000.0,0.7356916666666667)
(700000.0,1.0678685714285714)
(800000.0,0.9999887500000001)
(900000.0,0.9481044444444445)
(1000000.0,0.907695)
(1100000.0,0.8757845454545455)
(1200000.0,0.8556541666666667)
(1300000.0,0.9137799999999999)
(1400000.0,1.2433628571428572)
(1500000.0,1.196496)
(1600000.0,1.156124375)
(1700000.0,1.1211582352941176)
(1800000.0,1.0901633333333334)
(1900000.0,1.0627415789473684)
(2000000.0,1.0386235)
(2100000.0,1.0172309523809522)
(2200000.0,0.9980318181818182)
(2300000.0,0.9810021739130435)
(2400000.0,0.97435125)
(2500000.0,0.9602384)
(2600000.0,1.0353453846153846)
(2700000.0,1.452214074074074)
(2800000.0,1.4197203571428572)
(2900000.0,1.3896575862068967)
(3000000.0,1.3617313333333334)
(3100000.0,1.3357312903225809)
(3200000.0,1.311425)
(3300000.0,1.288671212121212)
(3400000.0,1.267429411764706)
(3500000.0,1.24752)
(3600000.0,1.2287325)
(3700000.0,1.2111502702702703)
(3800000.0,1.1946363157894737)
(3900000.0,1.1790884615384616)
(4000000.0,1.16431575)
(4100000.0,1.1503731707317073)
(4200000.0,1.137207857142857)
(4300000.0,1.1248327906976745)
(4400000.0,1.113095)
(4500000.0,1.1019664444444444)
(4600000.0,1.0914534782608698)
(4700000.0,1.08683)
(4800000.0,1.0773027083333333)
(4900000.0,1.0681142857142858)
(5000000.0,1.0593748)
(5100000.0,1.1425605882352943)
(5200000.0,1.1328905769230768)
(5300000.0,1.609740754716981)
(5400000.0,1.5901179629629631)
(5500000.0,1.5711907272727272)
(5600000.0,1.552995)
(5700000.0,1.5355473684210525)
(5800000.0,1.5187255172413794)
(5900000.0,1.50247593220339)
(6000000.0,1.4867811666666666)
(6100000.0,1.4716659016393443)
(6200000.0,1.4571098387096777)
(6300000.0,1.4430355555555556)
(6400000.0,1.42934203125)
(6500000.0,1.4162283076923077)
(6600000.0,1.4034863636363637)
(6700000.0,1.3910988059701492)
(6800000.0,1.3790963235294118)
(6900000.0,1.3674514492753624)
(7000000.0,1.3561341428571427)
(7100000.0,1.3451478873239437)
(7200000.0,1.3344718055555558)
(7300000.0,1.3240971232876713)
(7400000.0,1.3139735135135138)
(7500000.0,1.3041841333333333)
(7600000.0,1.2946911842105264)
(7700000.0,1.2855332467532468)
(7800000.0,1.2765382051282053)
(7900000.0,1.2678249367088608)
(8000000.0,1.259329)
(8100000.0,1.2510665432098766)
(8200000.0,1.243022317073171)
(8300000.0,1.2351932530120482)
%  p = 0.9
     } node[pos=0] (endofplot90) {};
%     \node [above,color=blue] at (endofplot90) {\tiny $p=0.9$};

\addplot [blue,dashed] coordinates {
(100000,0.39734)
(200000,0.431705)
(300000,0.42267333333333335)
(400000,0.4833075)
(500000,0.473566)
(600000,0.481245)
(700000,0.5579471428571429)
(800000,0.5500975)
(900000,0.5450655555555556)
(1000000,0.541962)
(1100000,0.5400854545454545)
(1200000,0.5490358333333333)
(1300000,0.6362592307692307)
(1400000,0.6277928571428572)
(1500000,0.6208946666666667)
(1600000,0.615255625)
(1700000,0.6107029411764706)
(1800000,0.6070155555555555)
(1900000,0.6041526315789474)
(2000000,0.601915)
(2100000,0.6002785714285714)
(2200000,0.5991759090909091)
(2300000,0.5985286956521739)
(2400000,0.6035383333333333)
(2500000,0.6034036)
(2600000,0.6895588461538461)
(2700000,0.6835577777777778)
(2800000,0.6780882142857143)
(2900000,0.6731213793103449)
(3000000,0.6686323333333334)
(3100000,0.6645090322580646)
(3200000,0.6607409375)
(3300000,0.6573260606060606)
(3400000,0.6542161764705883)
(3500000,0.6513865714285715)
(3600000,0.6487997222222223)
(3700000,0.6464605405405406)
(3800000,0.6443336842105263)
(3900000,0.6424110256410256)
(4000000,0.6406795)
(4100000,0.6391368292682926)
(4200000,0.6377557142857143)
(4300000,0.6365151162790698)
(4400000,0.6354188636363637)
(4500000,0.6344802222222222)
(4600000,0.633654347826087)
(4700000,0.6377382978723404)
(4800000,0.6370410416666666)
(4900000,0.6364532653061225)
(5000000,0.6358874)
(5100000,0.7249684313725491)
(5200000,0.7211267307692307)
(5300000,0.7174583018867925)
(5400000,0.7139570370370371)
(5500000,0.7106112727272728)
(5600000,0.7074108928571429)
(5700000,0.7043533333333334)
(5800000,0.7014298275862069)
(5900000,0.6986218644067796)
(6000000,0.6959298333333334)
(6100000,0.6933483606557377)
(6200000,0.6908704838709677)
(6300000,0.6884850793650794)
(6400000,0.686191875)
(6500000,0.683988)
(6600000,0.6818581818181818)
(6700000,0.6798126865671642)
(6800000,0.6778841176470588)
(6900000,0.6760121739130435)
(7000000,0.6741978571428572)
(7100000,0.6724328169014084)
(7200000,0.6707377777777778)
(7300000,0.6690934246575343)
(7400000,0.6675016216216216)
(7500000,0.6659670666666667)
(7600000,0.6644689473684211)
(7700000,0.6629892207792207)
(7800000,0.6615711538461538)
(7900000,0.6601637974683544)
(8000000,0.6588095)
(8100000,0.6574765432098766)
(8200000,0.6561834146341463)
(8300000,0.6549253012048193)
% p=0.9 "early rehash" at 80%
};
     
             \addplot  [red!70!white]  coordinates {
(100000.0,0.51785)
(200000.0,0.5570050000000001)
(300000.0,0.70665)
(400000.0,0.6336400000000001)
(500000.0,0.62458)
(600000.0,0.832715)
(700000.0,0.7818700000000001)
(800000.0,0.7413337500000001)
(900000.0,0.7117155555555557)
(1000000.0,0.729006)
(1100000.0,1.0215845454545456)
(1200000.0,0.9772008333333334)
(1300000.0,0.9402346153846153)
(1400000.0,0.914657142857143)
(1500000.0,0.887938)
(1600000.0,0.8646987500000001)
(1700000.0,0.8446894117647059)
(1800000.0,0.8272800000000001)
(1900000.0,0.8590715789473685)
(2000000.0,0.8432085)
(2100000.0,1.1755514285714286)
(2200000.0,1.2062981818181822)
(2300000.0,1.1760300000000001)
(2400000.0,1.1485191666666668)
(2500000.0,1.1234312)
(2600000.0,1.1004242307692307)
(2700000.0,1.0792848148148149)
(2800000.0,1.0664714285714287)
(2900000.0,1.0484641379310344)
(3000000.0,1.0318126666666667)
(3100000.0,1.0163106451612904)
(3200000.0,1.0019184375)
(3300000.0,0.9886966666666668)
(3400000.0,0.9763502941176471)
(3500000.0,0.9650237142857143)
(3600000.0,0.9544294444444446)
(3700000.0,0.9980218918918919)
(3800000.0,0.9870323684210526)
(3900000.0,0.9767584615384616)
(4000000.0,0.967172)
(4100000.0,1.0088480487804878)
(4200000.0,1.3243516666666666)
(4300000.0,1.4253541860465118)
(4400000.0,1.4048236363636366)
(4500000.0,1.385255111111111)
(4600000.0,1.3665621739130436)
(4700000.0,1.3487146808510637)
(4800000.0,1.3316064583333334)
(4900000.0,1.315261224489796)
(5000000.0,1.2996308)
(5100000.0,1.2846558823529413)
(5200000.0,1.270255576923077)
(5300000.0,1.2564650943396227)
(5400000.0,1.243209814814815)
(5500000.0,1.2305018181818181)
(5600000.0,1.2226532142857143)
(5700000.0,1.2108982456140351)
(5800000.0,1.1996308620689655)
(5900000.0,1.188775593220339)
(6000000.0,1.1783423333333334)
(6100000.0,1.168324262295082)
(6200000.0,1.1586153225806453)
(6300000.0,1.149248253968254)
(6400000.0,1.14024828125)
(6500000.0,1.1315279999999999)
(6600000.0,1.1231154545454547)
(6700000.0,1.1149652238805972)
(6800000.0,1.1071314705882354)
(6900000.0,1.0995278260869565)
(7000000.0,1.0921678571428572)
(7100000.0,1.0850487323943663)
(7200000.0,1.0781859722222225)
(7300000.0,1.12753)
(7400000.0,1.120220810810811)
(7500000.0,1.1131473333333333)
(7600000.0,1.106286052631579)
(7700000.0,1.0996411688311691)
(7800000.0,1.093194358974359)
(7900000.0,1.0869637974683546)
(8000000.0,1.080874125)
(8100000.0,1.1285965432098766)
(8200000.0,1.122195487804878)
(8300000.0,1.1669998795180725)
% p = 0.75
} node[pos=0] (endofplot75) {};
%     \node [above,color=red] at (endofplot75) {\tiny $p=0.75$};

\addplot [red!70!white,dashed] coordinates {
(100000,0.34741)
(200000,0.363545)
(300000,0.39518333333333333)
(400000,0.39716)
(500000,0.444382)
(600000,0.4396933333333333)
(700000,0.44498714285714286)
(800000,0.44568)
(900000,0.5023688888888889)
(1000000,0.499755)
(1100000,0.4983390909090909)
(1200000,0.49748416666666667)
(1300000,0.4974623076923077)
(1400000,0.5052564285714286)
(1500000,0.5056166666666667)
(1600000,0.50619375)
(1700000,0.5068458823529411)
(1800000,0.5671927777777778)
(1900000,0.5637757894736842)
(2000000,0.5609165)
(2100000,0.5585447619047619)
(2200000,0.556565)
(2300000,0.5549704347826087)
(2400000,0.5537025)
(2500000,0.5526916)
(2600000,0.5518626923076924)
(2700000,0.5512422222222222)
(2800000,0.5560110714285714)
(2900000,0.555641724137931)
(3000000,0.5554386666666666)
(3100000,0.5554335483870968)
(3200000,0.555589375)
(3300000,0.5558951515151516)
(3400000,0.5563455882352941)
(3500000,0.5569377142857143)
(3600000,0.6169730555555556)
(3700000,0.6141078378378378)
(3800000,0.6114363157894737)
(3900000,0.6089820512820513)
(4000000,0.60669175)
(4100000,0.6045919512195121)
(4200000,0.602652619047619)
(4300000,0.6008323255813953)
(4400000,0.5991597727272727)
(4500000,0.5975991111111111)
(4600000,0.5961458695652174)
(4700000,0.5948065957446809)
(4800000,0.5935766666666666)
(4900000,0.5924512244897959)
(5000000,0.5914146)
(5100000,0.5904574509803922)
(5200000,0.5895782692307693)
(5300000,0.5887705660377358)
(5400000,0.5880133333333334)
(5500000,0.5873343636363636)
(5600000,0.5910348214285714)
(5700000,0.5904585964912281)
(5800000,0.5899205172413793)
(5900000,0.5894313559322034)
(6000000,0.5889908333333334)
(6100000,0.5886101639344262)
(6200000,0.5882909677419355)
(6300000,0.588061746031746)
(6400000,0.58785296875)
(6500000,0.5877067692307693)
(6600000,0.5875928787878788)
(6700000,0.5875)
(6800000,0.5874833823529412)
(6900000,0.5874934782608696)
(7000000,0.5875208571428572)
(7100000,0.6489343661971831)
(7200000,0.6470605555555555)
(7300000,0.6452513698630137)
(7400000,0.6435148648648649)
(7500000,0.6418476)
(7600000,0.6402523684210526)
(7700000,0.6387097402597403)
(7800000,0.6372388461538462)
(7900000,0.6358658227848101)
(8000000,0.63448625)
(8100000,0.6331307407407407)
(8200000,0.6318148780487804)
(8300000,0.630558313253012)
% p=0.75, early rehash to 80%
};

\end{axis}
\end{tikzpicture}
\caption{Experimental plots for time against number of operations. The plot on the right shows the average microseconds/operation made on a personal computer. Blue plots correspond to a probability of insertion $p=0.9$, while red plots correspond to $p=0.75$. The solid plots correspond to the original \lua code, while the dashed plots correspond to the modified \Lua, in order to ensure free room when rehashing as described below (our fix proposal). The \Lua code for the experiments can be found at \url{https://gist.github.com/PRotondo/59036a5a3ddd53b30d9555a3a748cb7c}.}
    \label{fig:experiments-time}
    \vspace{-.5cm}
\end{figure}

In Section~\ref{sec:hybrid-tables}, we
have shown that the hybrid structure introduces similar issues  (see Prop.~\ref{worst-case-permutations-1}) when
considering only insertions. The effect is more limited than in the case
of both insertions and deletions (see Prop.~\ref{worst-case-permutations-2} and Thm.~\ref{thm:average-case-permutations2}), yet the array-part
might not be exploited to reduce memory consumption as much as would be expected.

These problems seem to have easy fixes, the most immediate one being to allow more room when rehashing, to avoid restarting with a new full or almost full table. A second solution would be to implement true deletions instead of just marking the deleted elements by setting their values to \nil. Both solutions are very classical and details can be found, for instance, in~\cite{Knuth3}.

We conclude by briefly considering the first fix.  
In practice, it is enough to change just one line of code. The function {\tt setnodesize} in {\tt ltable.c} creates a new hash-part having at least {\tt size}  elements, which is passed on as an argument. There, the new exponent $m$ of the size of the hash-part, called {\tt lsize} in the {\tt C} code, is chosen to be $\lceil\log_2{\tt size}\rceil$, in {\tt C} code  {\tt luaO\_ceillog2(size)}. To enforce extra space, it suffices to set {\tt lsize} to {\tt luaO\_ceillog2(size+(size>>2))}. Increasing the capacity of the new hash-part this way ensures that at least 20\% of its cells are free\footnote{More precisely, as {\tt size>>2} corresponds to $\lfloor {\tt size}/4\rfloor$ rather than $\lceil {\tt size}/4\rceil$, the number of free cells is guaranteed to be at least 20\% of the total, minus one extra cell due to the floor function.} after reinserting the old elements. The effects of this new rehash policy can be seen in Figure~\ref{fig:experiments-time}. This fix guarantees amortized constant time per operation: indeed,  between two consecutive rehashes we must perform at least $M/5$ insertions\footnote{Not $M/5-1$ because we need to perform one extra insertion to trigger the rehash.}. 

\medskip

\begin{wrapfigure}{l}{0.37\textwidth}
    \vspace{-.7cm} % TODO Change as needed
    \centering
\begin{lstlisting}
m = 1<<24
tab = {}

for i=(-m+1),m do
  tab[i] = 1
end
\end{lstlisting}
    \caption{\Lua code for the worst-case (insertions) in Example~\ref{example:hybrid-1}.}
    \label{code:hybrid-1}
\end{wrapfigure}

We also considered the effect of this modification on the worst-case scenarios described previously.
For the situation in Example~\ref{worst-case-hash} with $M=2^{15}$,
our simulations in a personal computer\footnote{Processor AMD Ryzen 5 PRO 2500U, 16 GB RAM, running LinuxMint. The \Lua code for the tests can be found in \url{https://gist.github.com/PRotondo/275a1292cd0b4cc08064211f2e600dc1}} yield a time of approximately 32 seconds for the original \Lua, while the modified Lua takes just 16 milliseconds.
Considering only insertions, \Lua takes 21 seconds for Example~\ref{example:hybrid-1} with $M=2^{24}$, 
whose five-line code can be seen in Figure~\ref{code:hybrid-1} on the left, against 3.5 seconds for the modified version.

\clearpage
\bibliography{lua-hashmaps}

\end{document}